\documentclass[12pt]{article}%

\usepackage{amssymb}
\usepackage{amsfonts}
\usepackage{amsmath}
\usepackage[nohead]{geometry}
\usepackage[singlespacing]{setspace}
\usepackage{indentfirst}
\usepackage{graphicx}%
\usepackage{rotating}
\usepackage{verbatim}
\usepackage{setspace}
\usepackage{multirow}
\usepackage{latexsym}

\usepackage{amsthm}
\usepackage{makeidx}
\usepackage{fancyhdr}

\newcommand{\hV}{\widehat  V }

\newcommand{\hu}{\widehat u}

\newcommand{\hSig}{\widehat\Sigma}

\newcommand{\hlam}{\widehat\lambda}
\newcommand{\hLam}{\widehat \Lambda}

\newcommand{\cov}{\mathrm{cov}}

\newcommand{\tr}{\mathrm{tr}}

\newcommand{\hO}{\widehat \Omega}
\newcommand{\var}{\mathrm{var}}
\newcommand{\beq}{\begin{eqnarray*}}
\newcommand{\eeq}{\end{eqnarray*}}

\newcommand{\Sigu}{\Sigma_{u}}
\newcommand{\Sigun}{\Sigma_{u0}}
\newcommand{\Siguni}{\Sigma_{u0}^{-1}}
\newcommand{\lamj}{\lambda_{0j}}
\newcommand{\hSigo}{\hSig_u^{(1)}}
\newcommand{\hSigoi}{(\hSigo)^{-1}}

\newcommand{\hLamo}{\hLam^{(1)}}
\newcommand{\hLamop}{\hLam^{(1)'}}
\newcommand{\hfto}{\widehat {f}_t^{(1)}}
\newcommand{\hftt}{\widehat {f}_t^{(2)}}
\newcommand{\fpca}{\widehat {f}_t^{PCA}}
\newcommand{\lampca}{\hlam^{PCA}}
\newcommand{\hLamt}{\hLam^{(2)}}
\newcommand{\hLamtp}{\hLam^{(2)'}}
\newcommand{\hSigt}{\hSig_u^{(2)}}
\newcommand{\hSigti}{(\hSigt)^{-1}}

\numberwithin{equation}{section}
\theoremstyle{plain}
\newtheorem{thm}{Theorem}[section]

\newtheorem{lem}{Lemma}[section]

\newtheorem{assum}{Assumption}[section]
\theoremstyle{definition}

\newtheorem{remark}{Remark}[section]

\makeatletter
\def\@biblabel#1{\hspace*{-\labelsep}}
\makeatother
\begin{document}

\title{ Efficient Estimation of Approximate Factor Models via Regularized Maximum Likelihood}
\author{Jushan Bai\medskip\\{\normalsize Columbia University}\\{\normalsize Department of Economics}      \and Yuan Liao \medskip\\{\normalsize  University of Maryland}\\{\normalsize Department of Mathematics}    }
\date{\today}
\maketitle

\sloppy%

\onehalfspacing


\begin{abstract}

We study the estimation of a high dimensional approximate factor model in the presence of both cross sectional dependence and heteroskedasticity. The classical   method of principal components analysis (PCA) does not efficiently estimate the factor loadings or common factors because it essentially treats the idiosyncratic error to be homoskedastic and cross sectionally uncorrelated.    For  efficient estimation it is   essential to estimate a large error covariance matrix. We assume the model to be conditionally sparse, and propose two approaches to estimating the common factors and factor loadings; both are based on maximizing a Gaussian quasi-likelihood and involve regularizing a large covariance sparse matrix. In the first approach the factor loadings and the error covariance are estimated separately while in the second approach they are estimated jointly.  Extensive asymptotic analysis has been carried out. In particular, we develop the inferential theory for the two-step estimation. Because the proposed approaches take into account the large error covariance matrix, they produce more efficient estimators than the classical PCA methods or methods based on a strict factor model.

\end{abstract}

\strut


\strut

\pagebreak%
\doublespacing

\onehalfspacing

\section{Introduction}
In many applications of economics, finance, and other scientific fields, researchers often face a large panel data set in which there are multiple observations for each individual; here individuals can be families, firms, countries, etc.  Modern applications usually involve data-rich environments in which both the number of observations for each individual and the number of individuals are  large.  One useful method for summarizing information in a large dataset is the factor model:
\begin{equation}
y_{it}=\alpha_i+\lambda_{0i}'f_t+u_{it}, \quad i\leq N, t\leq T, 
\end{equation}
where $\alpha_i$ is an individual effect, $\lambda_{0i}$ is an $r\times 1$ vector of factor loadings and $f_t$ is an  $r\times 1$ vector of common factors; $u_{it}$ denotes the idiosyncratic component of the model.  Note that $y_{it}$ is the only observable random variable in this model.  If we  write $y_t=(y_{1t},...,y_{Nt})'$, $\Lambda_0=(\lambda_{01},...,\lambda_{0N})'$, $\alpha=(\alpha_1,...,\alpha_N)'$ and $u_t=(u_{1t},...,u_{Nt})'$, then model (1.1) can be equivalently written as
$$
y_t=\alpha+\Lambda_0f_t+u_{t}.
$$

Because $y_{it}$ is the only observable in the model, both factors and loadings are treated as parameters to   estimate.  As  was shown by Chamberlain and Rothschild (1983),  in many applications of factor analysis, it is  desirable to allow  dependence  among the error terms $\{u_{it}\}_{i\leq N, t\leq T}$ not only serially but also cross-sectionally. This gives rise to the \textit{approximate factor model}, in which    the $N\times N$ covariance matrix $\Sigun=\cov(u_t)$ is not diagonal. In addition, the diagonal entries may vary in a large range. As a result, efficiently estimating the factor model under both large $N$ and large $T$ is difficult because to take into  account both cross-sectional heteroskedasticity and dependence of $\{u_{it}\}_{i\leq N, t\leq T}$, it is essential to estimate the large covariance $\Sigun$. The latter has been known as  a challenging problem when $N$ is larger than $T$.

In this paper, we assume the model to be \textit{conditionally sparse}, in the sense that $\Sigun$ is a sparse matrix with bounded eigenvalues. This assumption effectively reduces the number of parameters to be estimated in the model, and allows  a consistent estimation of $\Sigun$. The latter is needed to efficiently estimate the factor loadings. In addition,  it enables the model to identify the common components $\alpha_i+\lambda_{0i}'f_t$   asymptotically as $N\rightarrow\infty$.     We propose two alternative methods, both are likelihood-based. The first one is a two-step procedure. In step one, we   apply the \textit{principal orthogonal complement thresholding} (POET) estimator of Fan et al. (2012) to estimate $\Sigun$ using the adaptive thresholding as in Cai and Liu (2011); in step two, we estimate the factor loadings by maximizing a Gaussian-quasi likelihood function, which depends on the covariance estimator in the first step. These two steps can be carried out iteratively. We also propose an alternative method for jointly estimating the factor loadings and the error covariance matrix by maximizing a weighted $l_1$ penalized likelihood function. The likelihood penalizes the estimation of the off-diagonal entries of the error covariance and automatically produces  a sparse covariance estimator. We present asymptotic analysis for both methods. In particular, we derive the uniform rate of convergence and limiting distribution of the estimators for the two-step procedure.  The analysis of the joint-estimation is more difficult as it involves penalizing a large covariance with diverging eigenvalues.  We establish the consistency  for this method. 
 
Moreover, we achieve the ``sparsistency"     for the estimated error covariance matrix in factor analysis (see Section 3 for detailed explanations). The estimated covariance  is consistent for both approaches under the normalized Frobenius  norm even when $N$ is much larger than $T$. This is important in the applications of approximate factor models.

There has been a large  literature  on estimating the approximate factor model. Stock and Watson (1998, 2002) and Bai (2003) considered the principal components analysis (PCA), and they  developed large-sample inferential theory. However, the PCA essentially treats $u_{it}$ to have the same variance across $i$, hence is inefficient when cross-sectional heteroskedasticity is  present.    Choi (2012) proposed a generalized PCA that requires $N<T$ to invert  the error sample  covariance  matrix.    More recently, Bai and Li (2012)   estimated the factor loadings by maximizing the Gaussian-quasi likelihood, which  addresses the  heteroskedasticity under large $N$, but they consider the strict factor model in which $(u_{1t},...,u_{Nt})$ are uncorrelated.   Additional literature on factor analysis includes, e.g., Bai  and Ng (2002),  Wang (2009), Dias, Pinherio and Rua (2008), Breitung and Tenhofen (2011), Han (2012), etc;  most of these studies are based on the PCA method. In contrast, our methods are maximum-likelihood-based.    Maximum likelihood methods have been one of the fundamental tools for statistical estimation and inference.


Our approach is closely related to the large covariance estimation literature, which has been rapidly growing in recent years.   There are in general two ways to estimate a sparse covariance in the literature:  thresholding and penalized maximum likelihood.  For our two-step procedure, we apply the POET estimator recently proposed by Fan et al. (2012), corresponding to the thresholding approach of Bickel and Levina (2008a), Rothman et al. (2009) and Cai and Liu (2011). For the joint estimation procedure, we use the penalized likelihood, corresponding to that of Lam and Fan (2009), Bien and Tibshirani (2011), etc. In either way, we need to show that the impact of estimating the large covariances is asymptotically negligible for an efficient estimation, which is not easy in our context since the likelihood function is highly nonlinear, and   $\Lambda_0\Lambda_0'$ contains a few eigenvalues that grow very fast. It was recently shown by Fan et al. (2012) that estimating a covariance matrix with fast diverging eigenvalues is a challenging problem. Other works on large covariance estimation include  Cai and Zhou (2012),   Fan et al. (2008), Jung and Marron (2009), Witten, Tibshirani and Hastie (2009),  Deng and Tsui (2010), Yuan (2010), Ledoit and Wolf (2012), El Karoui (2008), Pati et al. (2012), Rohde and Tsybakov (2011),  Zhou et al. (2011),   Ravikumar et al. (2011) etc.

This paper focuses on  high-dimensional \textit{static   factor models} although  the factors and errors can be serially correlated. The model considered is different from the generalized  \textit{dynamic factor models} as in  Forni, Hallin, Lippi and Reichlin (2000),  Forni and Lippi (2001),  Hallin and Li\v{s}ka (2007), and other references therein. Both static and  dynamic factor models are receiving increasing attention in  applications of many fields.

The paper is organized as follows. Section 2 introduces  the conditional sparsity assumption and  the likelihood function. Section 3 proposes the two-step estimation procedure. In particular, we present asymptotic inferential theory of the estimators. Both uniform rate of convergence and limiting distributions are derived. Section 4 gives the joint estimation as an alternative procedure, where  we demonstrate the estimation consistency. Section 5 illustrates some numerical examples which compare the proposed methods with the existing ones in the literature. Finally, Section 6 concludes with further discussions. All proofs are given in the appendix.

\textbf{Notation}

Let $\lambda_{\max}(A)$ and $\lambda_{\min}(A)$ denote the maximum and minimum eigenvalues of a matrix $A$ respectively. Also Let $\|A\|_1$, $\|A\| $ and $\|A\|_F$ denote the $l_1$, spectral   and Frobenius norms of $A$, respectively. They are defined  as $\|A\|_1=\max_{i}\sum_j|A_{ij}|$, $\|A\| =\sqrt{\lambda_{\max}(A'A)}$, $\|A\|_F=\sqrt{\tr(A'A)}$.  Note that $\|A\|$ is also the Euclidean norm when $A$ is a vector. For two sequences $a_T$ and $b_T$, we write $a_T\ll b_T$, and equivalently $b_T\gg a_T$, if $a_T=o(b_T)$ as $T\rightarrow\infty.$

\section{Approximate Factor Models}
\subsection{The model}
 The approximate factor  model (1.1) implies the following covariance decomposition:
\begin{equation}\label{eq2.2}
\Sigma_{y0}=\Lambda_0\,\cov(f_t)\,\Lambda_0'+\Sigun,
\end{equation}
 assuming $f_t$ to be uncorrelated with $u_t$,
where $\Sigma_{y0}$ and $\Sigun$ denote the $N\times N$ covariance matrices of $y_t$ and $u_t$; $\cov(f_t)$ denotes the $r\times r$ covariance of $f_t$, all assumed to be time-invariant. The approximate factor model typically requires the idiosyncratic covariance $\Sigun$ have bounded eigenvalues and $\Lambda_0'\Lambda_0$ have eigenvalues diverging at rate $O(N)$. One of the key concepts of approximate factor models is that it allows $\Sigma_{u0}$ to be non-diagonal.
 
Stock and Watson (1998) and Bai (2003) derived the rates of convergence as well as the inferential theory of  the method of principal component analysis (PCA) for estimating the factors and loadings.  Let $Y=(y_1,..., y_T)'$ be the  $T\times N$ data matrix. Then PCA estimates the $T\times r$ factor matrix $F$ by maximizing $\tr(F'(YY')F)$ subject to normalization restrictions for $F$.  The PCA method  essentially restricts to have cross-sectional homoskedasticity and independence.  Thus it is known to be inefficient when the idiosyncratic errors are either cross sectionally heteroskedastic or correlated. 

 This paper aims at the efficient estimation of the approximate factor model, and assumes the number of factors $r$ to be known.  In practice, $r$ can be estimated from the data, and there has been a large literature addressing its consistent estimation, e.g., Bai and Ng (2002), Kapetanios (2010), Onatski (2010),  Alessi et al. (2010), Hallin and Li\v{s}ka (2007), Lam and Yao (2012), among others.
 
 \subsection{Conditional sparsity}
 An efficient estimation of the factor loadings and factors should  take into account both cross-sectional dependence and heteroskedasticity, which will then involve estimating   $\Sigma_{u0}=\cov(u_t)$, or more precisely, the precision matrix $\Sigma_{u0}^{-1}$. In a data-rich environment, $N$ can be either comparable with or much larger than $T$. Then estimating $\Sigma_{u0}$ is a challenging problem even when the idiosyncratics $\{u_{it}\}_{i\leq N, t\leq T}$ are observable, because the sample covariance is nonsingular when $N>T$, whose spectrum is inconsistent (Johnstone and Ma 2009).

Under the    regular approximate factor model considered by Chamberlain and Rothschild (1983) and Stock and Watson (2002), it is difficult to estimate $\Sigun$ without further structural assumptions. A natural assumption to go one-step further is   that of sparsity, which assumes  that many off-diagonal elements of $\Sigun$ be either zero or vanishing as the dimensionality increases. In an approximate factor model,     it is more appropriate to assume $\Sigun$ be a sparse matrix instead of $\Sigma_{y0}$. Due to the presence of common factors, we call such a special structure of the factor model to be \textit{conditionally sparse}.
  
  Therefore, the model studied in the current paper is the approximate factor model with conditional sparsity (sparsity structure on $\Sigma_{u0}$), which is  sightly more restrictive than that of  Chamberlain and Rothschild (1983). The conditional sparsity is required to regularize a large idiosyncratic covariance, which allows us to take both cross sectional correlation and heteroskedasticity into account, and is needed for an efficient estimation. However, such an assumption is still quite general and covers most of the applications of factor models in economics, finance, genomics, and many important applied areas.

\subsection{Maximum likelihood}
Compared to PCA, a more efficient estimation for model (2.1) of high dimension  is based on a Gaussian quasi-likelihood approach.  Let $\bar{f}=T^{-1}\sum_{t=1}^Tf_t.$ Because of the existence of $\alpha$, the model $y_t=\Lambda_0f_t+\alpha+u_t$ is observationally equivalent to $y_t=\Lambda_0f_t^*+\alpha^*+u_t$, where $f_t^*=f_t-\bar{f}$ and $\alpha^*=\alpha+\Lambda_0\bar{f}.$ Therefore without loss of generality,  we assume  $\bar{f}=0$. The Guassian quasi-likelihood for $\Sigma_y$ is given by $$-N^{-1}\log|\det(\Sigma_y)|-N^{-1}\tr(S_y\Sigma_y^{-1})$$ where $S_y=T^{-1}\sum_{t=1}^T(y_t-\bar{y})(y_t-\bar{y})'$ is the sample covariance matrix, with $\bar{y}=T^{-1}\sum_{t=1}^Ty_t$. Plugging in (\ref{eq2.2}), using the notation  $S_f=\frac{1}{T}\sum_{t=1}^Tf_tf_t'$, we obtain the quasi-likelihood function for the factors and loadings:
\begin{equation}\label{eq2.3}
-\frac{1}{N}\log\left|\det\left(\Lambda S_f\Lambda'+\Sigma_u\right)\right|-\frac{1}{N}\tr\left(S_y(\Lambda S_f\Lambda'+\Sigu)^{-1}\right),
\end{equation}
where $\Lambda=(\lambda_1,...,\lambda_N)'$ is an $N\times r$ matrix of factor loadings. 

It has been well known that the factors and loadings are not separably identified without further restrictions. Note that the factors and loadings enter the likelihood through $\Lambda S_f\Lambda'$. Hence for any invertible $r\times r$ matrix $\bar{H}$, if we define $\Lambda^*=\Lambda \bar{H}^{-1}$,  $f_t^*=\bar{H}f_t$ and $S_{f^*}=\frac{1}{T}\sum_{t=1}^Tf^*_tf_t^{*'}$, then $\Lambda^*S_{f^*}\Lambda^{*'}=\Lambda S_f \Lambda'$, and they produce observationally equivalent models.    In this paper, we focus on a usual restriction for MLE of factor analysis (see e.g., Lawley and Maxwell 1971) as follows:
 \begin{equation}\label{eq2.4}
S_f=I_r,\text{ and } \Lambda'\Sigu^{-1}\Lambda \text{ is diagonal,}
\end{equation}and the diagonal entries of $ \Lambda'\Sigu^{-1}\Lambda$ are  distinct and are arranged in a decreasing order. Restriction (\ref{eq2.4}) guarantees a unique solution to the maximization of the log-likelihood function up to a column sign change for $\Lambda$. Therefore we assume the estimator $\hLam$ and $\Lambda_0$ have the same column signs, as part of the identification conditions.

The negative log-likelihood function  (\ref{eq2.3}) simplifies to 
\begin{equation}\label{eq2.5}
-L(\Lambda,\Sigu)=\frac{1}{N}\log\left|\det\left(\Lambda\Lambda'+\Sigma_u\right)\right|+\frac{1}{N}\tr\left(S_y(\Lambda\Lambda'+\Sigu)^{-1}\right).
\end{equation}


In the presence of cross sectional dependence, $\Sigun$ is not necessarily diagonal. Therefore there can be up to $O(N^2)$ free parameters in the  likelihood function (\ref{eq2.5}).  There are in general two main    regularization approaches to estimating a large sparse covariance:  (adaptive) thresholding (Bickel and Levina 2008a, Rothman et al. 2009, Cai and Liu 2011, etc.) and penalized maximum likelihood (Lam and Fan 2009, Bien and Tibshirani 2011). Correspondingly in this paper, we propose two methods for regularizing the likelihood function to efficiently estimate  the factor loadings as well as the unknown factors. One estimates $\Sigun$ and $\Lambda_0$ in two steps and the other estimates them jointly.

\section{Two-Step Estimation}

The two-step estimation estimates $(\Lambda_0,\Sigun)$ separately. In the first step, we  estimate $\Sigun$ by the principal orthogonal complement thresholding (POET), proposed by Fan et al. (2012), and in the second step we estimate $\Lambda_0$ only, using the quasi-maximum likelihood, replacing $\Sigma_u$ by the covariance estimator obtained in step one.

\subsection{Step one: covariance estimation by thresholding}
The  POET  is based on a spectrum expansion of the sample covariance matrix and adaptive thresholding. Let $(\nu_j,\xi_j)_{j=1}^N$ be the eigenvalues-vectors of the sample covariance $S_y$ of $y_t$, in a decreasing order such that $\nu_1\geq \nu_2\geq...\geq\nu_N.$ Then $S_y$ has the following spectrum decomposition:
$$
S_y=\sum_{i=1}^r\nu_i\xi_i\xi_i'+R
$$
where $R=\sum_{i=r+1}^N\nu_i\xi_i\xi_i'$ is the \textit{orthogonal complement component}. Define a general thresholding function $s_{ij}(z): \mathbb{R}\rightarrow\mathbb{R}$ as in Rothman et al. (2009) and Cai and Liu (2011)  with an entry-dependent threshold $\tau_{ij}$ such that:\\
(i) $s_{ij}(z)=0$  if $|z|<\tau_{ij};$\\
(ii) $|s_{ij}(z)-z|\leq \tau_{ij}.$\\
(iii) There are constants $a>0$ and $b>1$ such that $|s_{ij}(z)-z|\leq a\tau_{ij}^2$ if $|z|>b\tau_{ij}$.\\
Examples of $s_{ij}(z)$ include the hard-thresholding: $s_{ij}(z)=zI_{(|z|>\tau_{ij})}$; SCAD (Fan and Li 2001), MPC (Zhang 2010) etc. Then we obtain the step-one consistent estimator for $\Sigun$:
$$
\hSig_u^{(1)}=(s_{ij}(R_{ij}))_{N\times N}, \text{ where } R=(R_{ij})_{N\times N}.
$$
We can choose the threshold as $\tau_{ij}=C\sqrt{R_{ii}R_{jj}}(\sqrt{(\log N)/T}+1/\sqrt{N})$ for some universal constant $C>0$, which corresponds to applying the threshold $C(\sqrt{(\log N)/T}+1/\sqrt{N})$ to the correlation matrix of $R$ [defined to be diag$(R)^{-1/2}R$ diag$(R)^{-1/2}$]. The POET estimator also has an equivalent expression using PCA. Let $\{\hu_{it}^{\text{PCA}}\}_{i\leq N, t\leq T}$ denote the PCA estimators of $\{u_{it}\}_{i\leq N, t\leq T}$ (Bai 2003). Then $\widehat\Sigma_{u,ij}^{(1)}=s_{ij}(T^{-1}\sum_{t=1}^T\hu_{it}^{\text{PCA}}\hu_{jt}^{\text{PCA}})$.

It was shown by Fan et al. (2012) that under some regularity conditions $\|\hSigo-\Sigun\| =O_p(N^{-1/2}+T^{-1/2}(\log N)^{1/2})$, which guarantees the positive definiteness asymptotically, given that $\lambda_{\min}(\Sigun)>0$ is bounded away from zero.

\subsection{Step two: estimating factor loadings and factors}
Replacing $\Sigu$ in (\ref{eq2.5}) by $\hSigo$, we obtain the objective function for $\Lambda$. Under the identification condition (\ref{eq2.4}), in this step, we estimate the loadings as: 
\begin{eqnarray}\label{eq3.1}
\hLamo&=&\arg\min_{\Lambda\in\Theta_{\lambda}}L_1(\Lambda)\cr
&=&\arg\min_{\Lambda\in\Theta_{\lambda}}\frac{1}{N}\log|\det(\Lambda\Lambda'+\hSigo)|+\frac{1}{N}\tr(S_y(\Lambda\Lambda'+\hSigo)^{-1})
\end{eqnarray}
where $\Theta_{\lambda}$ is a parameter space for the loading matrix, to be defined later. Suppose that  $y_t\sim N(0,\Lambda_0\Lambda_0'+\Sigun)$, the negative log-likelihood is then  the same (up to a constant) as (\ref{eq3.1}) except that $\hSigo$ should be replaced by $\Sigun$. Consequently, (\ref{eq3.1}) can be treated as a Gaussian quasi-likelihood of $\Lambda$, which will give an efficient estimation of $\Lambda_0$ since it takes into account the cross sectional heteroskedasticity and dependence in $\Sigun$ through its consistent estimator.  

After obtaining $\hLamo$, we estimate $f_t$ via the generalized least squares (GLS) as suggested by Bai and Li (2012): 
$$
\hfto=(\hLamop(\hSigo)^{-1}\hLamo)^{-1}\hLamop(\hSigo)^{-1} (y_t-\bar{y}).
$$

The proposed two-step procedure can be carried out iteratively. After obtaining $(\hLamo, \hfto)$, we  update
$$
\hu_t=y_t-\hLamo\hfto, \quad \hSigo=(s_{ij}(T^{-1}\sum_{t=1}^T\hu_{it}\hu_{jt}))_{N\times N}.
$$
Then $\hSigo$ in the objective function (\ref{eq3.1}) is updated, which gives updated  $\hLamo$ and $\hfto$ respectively. This procedure can be continued until convergence.

\subsection{Positive definiteness}
The objective function (\ref{eq3.1})  requires $\Lambda\Lambda'+\hSigo$ be positive definite for any given finite sample. A sufficient condition is the finite-sample positive definiteness of $\hSigo$,
which also depends on the choice of the adaptive threshold value $\tau_{ij}$.  We specify
$$\tau_{ij}=C\alpha_{ij}\left(\sqrt{\frac{\log N}{T}}+\frac{1}{\sqrt{N}}\right)$$
where $\alpha_{ij}$ is an entry-dependent value that captures the variability of individual variables such as $\sqrt{R_{ii}R_{jj}}$; $C>0$ is a pre-determined universal constant. More concretely, the finite sample positive definiteness depends on the choice of $C.$ If we write $\hSigo=\hSigo(C)$  in step one to indicate its dependence on the threshold, then  $C$ should be chosen in the interval $(C_{\min}, C_{\max}]$, where $$
C_{\min}=\inf\{M: \lambda_{\min}(\hSigo(C))>0, \forall C>M\},
$$
and $C_{\max}$ is a large constant that thresholds all the off-diagonal elements of $\hSigo$ to zero. Then by construction, $\hSigo(C)$ is finite-sample positive definite for any $C>C_{\min}$ (see Figure \ref{f1}).
\begin{figure}[htbp]
\begin{center}
\caption{Minimum eigenvalue of $\lambda_{\min}(\hSigo(C)) $}
\includegraphics[width=10cm]{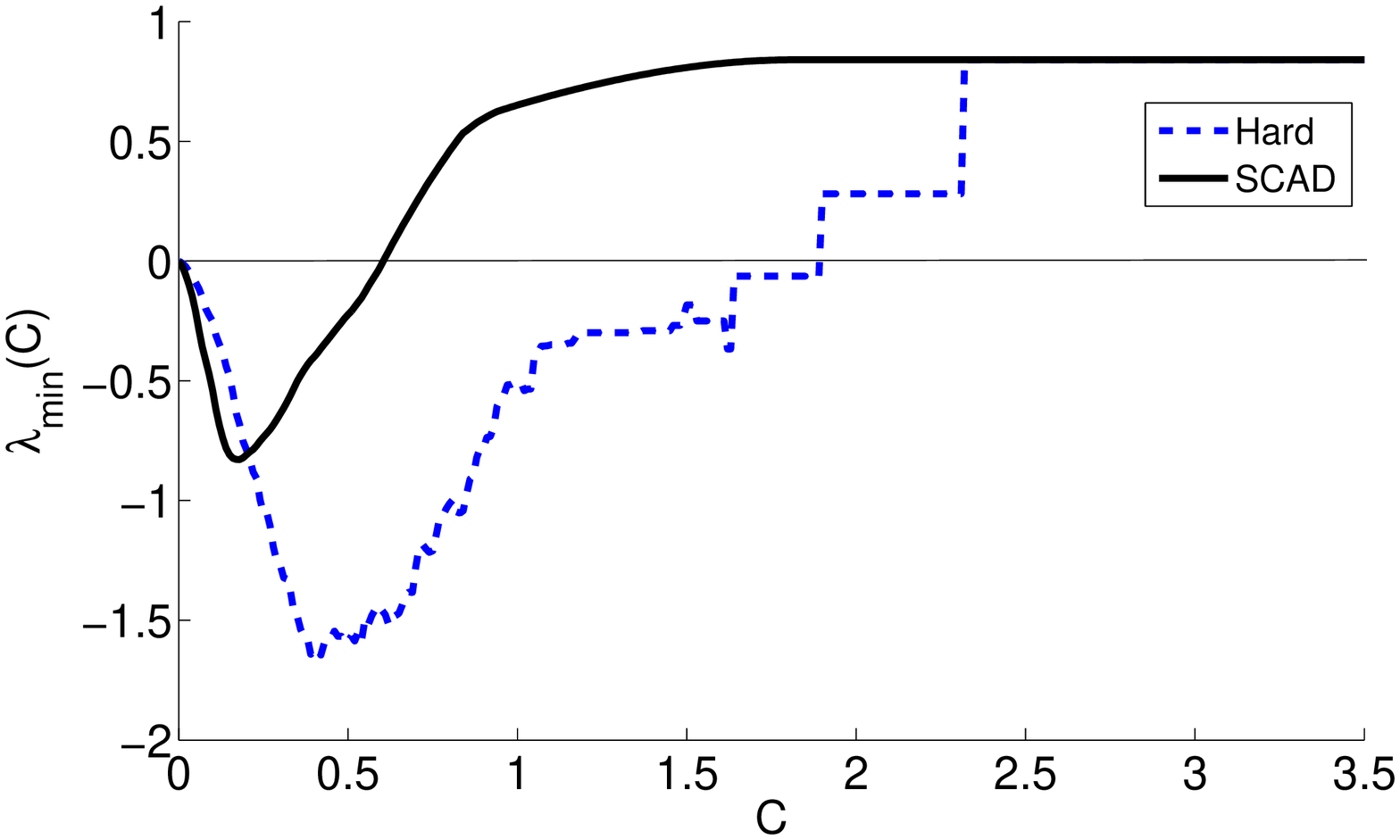}
 
\small Data are simulated from the setting of Section 5 with $T=100, N=150$. Both hard and SCAD with adaptive thresholds (Cai and Liu 2011) are plotted.
\label{f1}
\end{center}
\end{figure}

\subsection{Asymptotic analysis} 
We now present the asymptotic analysis of the proposed two-step estimator. We first list a set of regularity conditions and then present the consistency. A more refined set of assumptions are needed to achieve the optimal rate of convergence as well as the limiting distributions.

\subsubsection{Consistency}

\begin{assum}\label{ass3.1}
Let $\Sigma_{u0,ij}$ denote the $(i,j)$th entry of $\Sigun$.  There is  $q\in[0,1)$ such that
  $$
  m_N\equiv\max_{i\leq N}\sum_{j=1}^N|\Sigma_{u0,ij}|^q=o(\min(\sqrt{N}, \sqrt{T/\log N})).
  $$
In particular, when $q=0$, we define $m_N=\max_{i\leq N}\sum_{j=1}^NI_{(\Sigma_{u0,ij}\neq0)}$, which corresponds to the ``exactly sparse" case.
\end{assum}
The first assumption sets a condition on the sparsity of $\Sigun$,   under which Fan et al. (2012)  showed that the POET estimator $\hSigo$ is consistent under the operator norm.  The sparsity is in terms of the maximum row sum, considered by Bickel and Levina (2008a).

The following assumption provides the regularity conditions on the data generating process. We introduce the strong mixing condition. Let $\mathcal{F}_{-\infty}^0$ and $\mathcal{F}_{T}^{\infty}$ denote the $\sigma$-algebras generated by $\{(f_t,u_t): -\infty\leq t\leq 0\}$ and  $\{(f_t,u_t): T\leq t\leq \infty\}$ respectively. In addition, define the mixing coefficient
\begin{equation} \label{mixing}
\alpha(T)=\sup_{A\in\mathcal{F}_{-\infty}^0, B\in\mathcal{F}_{T}^{\infty}}|P(A)P(B)-P(AB)|.
\end{equation}

\begin{assum}\label{ass3.2} (i)
 $\{u_t, f_t\}_{t\geq1}$ is strictly stationary. In addition, $Eu_{it}=Eu_{it}f_{jt}=0$ for all $i\leq p, j\leq r$ and $t\leq T.$
\\
(ii) There exist  constants $c_1, c_2>0$   such that  $c_2<\lambda_{\min}(\Sigun)\leq\lambda_{\max}(\Sigun)<c_1,$ and $\max_{j\leq N}\|\lambda_{0j}\|<c_1$. \\
(iii)   Exponential tail:   There exist $r_1, r_2>0$ and $b_1, b_2>0$, such that for any $s>0$, $i\leq p$ and $j\leq r$,
\begin{equation*}
P(|u_{it}|>s)\leq\exp(-(s/b_1)^{r_1}), \quad P(|f_{jt}|>s)\leq \exp(-(s/b_2)^{r_2}).
\end{equation*}
(iv) Strong mixing: There exists   $r_3>0$ such that $3r_1^{-1}+1.5r_2^{-1}+r_3^{-1}>1$, and $C>0$ satisfying:  for all $T\in\mathbb{Z}^+$,
$$\alpha(T)\leq \exp(-CT^{r_3}).$$
\end{assum}

The following assumptions are standard in the approximate factor models, see e.g., Stock and Watson (1998, 2002) and Bai (2003). In particular, Assumption \ref{ass3.3} implies that the first $r$ eigenvalues of $\Lambda_0\Lambda_0'$ are growing rapidly at $O(N)$. Intuitively, it requires the factors be pervasive in the sense that  they  impact a non-vanishing
proportion of time series  $\{y_{1t}\}_{t\leq T},...,\{y_{Nt}\}_{t\leq T}$.

\begin{assum}\label{ass3.3}  There is a $\delta>0$ such that for all large $N$,   $$\delta^{-1}<\lambda_{\min}(N^{-1}\Lambda_0'\Lambda_0)\leq\lambda_{\max}(N^{-1}\Lambda_0'\Lambda_0)<\delta.$$
Therefore all the eigenvalues of $N^{-1}\Lambda_0'\Lambda_0$ are bounded away from both zero and infinity as $N\rightarrow\infty.$
\end{assum}
\begin{assum}\label{ass3.4}   There exists $M>0$  such that  for all    $t\leq T$ and $s\leq T$,\\
(i)  $E[N^{-1/2}({u}_s'{u}_t-E{u}_s'{u}_t)]^4<M$,\\
(ii) $E\|N^{-1/2}\sum_{j=1}^N\lamj u_{jt}\|^4<M$.
\end{assum}

 The following assumption defines the threshold $\tau_{ij}$  on the $(i,j)$th entry of $R_{ij}$ for the step-one POET estimator.  
\begin{assum}\label{ass3.5} The threshold $\tau_{ij}=C\alpha_{ij}(\sqrt{(\log N)/T}+1/\sqrt{T})$ where $\alpha_{ij}>0$ is  entry-dependent, either stochastic or deterministic, such that $\forall\epsilon>0$,  there are positive $C_1$ and $C_2$ so that
\begin{equation}\label{eq3.3}
P(C_1<\min_{i,j\leq N}\alpha_{ij}\leq\max_{i,j\leq N}\alpha_{i,j}<C_2)>1-\epsilon
\end{equation}for all large $N$ and $T$. Here $C>0$ is a deterministic constant. 
\end{assum}
 Condition (\ref{eq3.3}) requires the  rate $\tau_{ij}\asymp (\sqrt{(\log N)/T}+1/\sqrt{T})$ uniformly in $(i,j)$. This condition  is satisfied by the universal threshold $\alpha_{ij}=\alpha$ for all $(i,j)$, the  correlation threshold $\alpha_{ij}=\sqrt{R_{ii}R_{jj}}$ as discussed before, and the adaptive threshold in Cai and Liu (2011).

For identification, we require the objective function be minimized subject to the diagonality of $\Lambda'(\hSigo)^{-1}\Lambda$. In addition, since Assumption \ref{ass3.3} is essential in asymptotically identifying the covariance decomposition $\Sigma_{y0}=\Lambda_0\Lambda_0'+\Sigun$, we need to take it into account when minimizing the objective function. Therefore we assume $\delta$ in Assumption \ref{ass3.3} is sufficiently large, which leads to the following parameter space:
\begin{eqnarray}\label{eq3.4}
\Theta_{\lambda}=\{\Lambda: &&\delta^{-1}<\lambda_{\min}(N^{-1}\Lambda'\Lambda)\leq\lambda_{\max}(N^{-1}\Lambda'\Lambda)<\delta,\cr
  &&\Lambda'(\hSigo)^{-1}\Lambda \text{ is diagonal.}\}
\end{eqnarray}


Write $\gamma^{-1}=3r_1^{-1}+1.5r_2^{-1}+r_3^{-1}+1$ and $\hLamo=(\hlam_1^{(1)},...,\hlam_N^{(1)})'$. We have the following theorem.
\begin{thm}\label{th3.1} Suppose $(\log N)^{6/\gamma}=o(T),$ $T=o(N^2).$
Under Assumptions \ref{ass3.1}-\ref{ass3.5}, 
$$
\frac{1}{\sqrt{N}}\|\hLamo-\Lambda_0\|_F=o_p(1),\quad
\max_{j\leq N}\|\hlam_j^{(1)}-\lamj\|=o_p(1).
$$
\end{thm}

By a more careful large-sample analysis, we can improve the above result and derive the rate of convergence. Throughout the paper, we will frequently use the notation:
$$
\omega_T=\frac{1}{\sqrt{N}}+\sqrt{\frac{\log N}{T}}.
$$

\begin{thm}\label{th3.2}
Under the Assumptions of Theorem \ref{th3.1},
$$
\frac{1}{\sqrt{N}}\|\hLamo-\Lambda_0\|_F=O_p(m_N\omega_T^{1-q}), \quad
\max_{j\leq N}\|\hlam_j^{(1)}-\lamj\|=O_p(m_N\omega_T^{1-q}),
$$
where $m_N$ and $q$ are defined in Assumption \ref{ass3.1}.
\end{thm}
\begin{remark} In the above theorem $m_N$ does not need to be bounded. But in order to achieve the $\sqrt{T}$-consistency for each $\hlam_j$, the uniform rate of convergence above would require it be bounded (which is a strong assumption on the sparsity of $\Sigun$). Later in Section 3.4.3 we will enhance this  convergence rate  so that the boundedness of $m_N$ is not necessary and $\sqrt{T}$-consistency can still be achieved. This will require additional regularity conditions.

\end{remark}

\subsubsection{Covariance estimation and sparsistency}

In order to obtain the limiting distribution for each individual $\hlam_j^{(1)}$, we also need to achieve the \textit{sparsistency} for estimating $\Sigun.$ By sparsistency, we mean the property that all small entries of $\Sigun$   are  estimated as exactly zeros with  a probability  arbitrarily close to one. Besides being important for deriving the limiting distribution of $\hlam_j^{(1)}$, the  sparsistency itself is of independent interest for large covariance estimation, and  has been studied by many authors, for instance, Lam and Fan (2009) and Rothman et al. (2009).   To our best knowledge, this is the first place where the sparsistency for an estimated idiosyncratic $\Sigun$ is achieved in a high dimensional approximate factor model. 

 Let $S_L$ and $S_U$ denote two disjoint sets and  respectively include the indices of small and large elements of $\Sigun$ in absolute value, and
$$
\{(i,j): i\leq N, j\leq N\}=S_L\cup S_U.
$$
Because the diagonal elements represent the individual variances of the idiosyncratic components, we assume $(i,i)\in S_U$ for all $i\leq N.$  The   sparsity assumes that most of the indices $(i,j)$ belong to $S_L$ when $i\neq j$.  A special case arises when $\Sigma_{u0}$ is strictly sparse, in the sense that its elements in small magnitudes ($S_L$) are exactly zero. For the banded matrix as an example,
$$\Sigma_{u0,ij}\neq 0 \text{ if }|i-j|\leq k;\quad \Sigma_{u0,ij}=0 \text{ if }|i-j|> k$$
for some fixed $k.$ Then $S_L=\{(i,j):|i-j|>k\}$ and $S_U=\{(i,j):|i-j|\leq k\}$.

The following assumption quantifies the ``small" and ``big" entries of $\Sigun$. By ``small"  entries  we mean those of smaller order than $\omega_T=N^{-1/2}+T^{-1/2}(\log N)^{1/2}.$  The partition $\{(i,j): i\leq N, j\leq N\}=S_L\cup S_U$ may not be unique. Our analysis suffices as long as such a partition exists.
\begin{assum}\label{ass3.6}
There is a partition $\{(i,j): i\leq N, j\leq N\}=S_L\cup S_U$ such that $(i,i)\in S_U$ for all $i\leq N$ and $S_L$ is nonempty. In addition, 
$$\max_{(i,j)\in S_L}|\Sigma_{u0, ij}|\ll  \omega_T\ll\min_{(i,j)\in S_U}|\Sigma_{u0, ij}|.$$
\end{assum}

The conditional sparsity assumption requires most off-diagonal entries of $\Sigun$ be inside $S_L$, hence it is reasonable to have $S_L\neq\emptyset$ in the condition.  It is likely that $S_U$ only contains the diagonal elements. It then essentially corresponds to the strict factor model where $\Sigun$ is almost a diagonal matrix and error terms are only weakly cross-sectionally correlated.  That is also a special case of Assumption \ref{ass3.6}.


\begin{thm}\label{th3.3}
Under Assumption \ref{ass3.6} and those of Theorem \ref{th3.2}, for any $\epsilon>0$ and $M>0$, 
there is an integer $N_0>0$ such that as long as $T$ and $ N>N_0$,
\begin{eqnarray*}
 P(\hSig_{u,ij}^{(1)}=0, \forall (i,j)\in S_L)>1-\epsilon, \\
 P(|\hSig_{u,ij}^{(1)}|>M\omega_T, \forall (i,j)\in S_U)>1-\epsilon.
\end{eqnarray*}
\end{thm}

\vspace{2em}

  It was shown by Fan et al. (2012) that $\|\hSigoi-\Sigun^{-1}\|=O_p(m_N\omega_T^{1-q})$. Theorem \ref{l3.1} below demonstrates a strengthened convergence rate for the averaged estimation error.

\begin{assum} \label{ass3.7} There is $c>0$ such that $\|\Sigun^{-1}\|_1<c.$
\end{assum}

In addition to Assumptions \ref{ass3.1} and \ref{ass3.6}, we require the following condition on the sparsity of $\Sigun$, which further characterizes $S_L$ and $S_U$:

\begin{assum}\label{ass3.8add}
The index sets $S_L$ and $S_U$ satisfy: $\sum_{i\neq j, (i,j)\in S_U}1=O(N)$ and  
  $\sum_{(i,j)\in S_L}|\Sigma_{u0,ij}|=O(1)$.
\end{assum}
Assumption \ref{ass3.8add} requires that the number of off-diagonal large entries of $\Sigun$ be of order $O(N)$, and that the absolute sum of the small entries is bounded. This assumption is satisfied, for example, if $\{u_{it}\}_{i\leq N}$ follows an heteroskedastic MA($p$) process with a fixed $p$, where $\sum_{i\neq j, (i,j)\in S_U}1=O(N)$ and     $\sum_{(i,j)\in S_L}|\Sigma_{u0,ij}|=0$. It is also satisfied by banded matrices (Bickel and Levina 2008b, Cai and Yuan 2012) and block-diagonal matrices with fixed block size.

Define an $r\times N$ matrix $\Xi=\Lambda_0'\Siguni=(\xi_1,...,\xi_N)$. Then $\|\Sigun^{-1}\|_1<c$ implies 
$$\max_{j\leq N}\|\xi_j\|=\max_{j\leq N}\|\sum_{i=1}^N\lambda_{0i}(\Siguni)_{ij}\|\leq\|\Siguni\|_1\max_{j\leq N}\|\lamj\|<\infty.$$

The following assumption corresponds to those of PCA in Bai (2003), and also extends to the non-diagonal $\Sigun$.
\begin{assum}\label{ass3.9}
(i)
$E\|\frac{1}{\sqrt{TN}}\sum_{s=1}^Tf_s(u_s'u_t-Eu_s'u_t)\|^2=O(1)$\\
(ii) 
For each element $d_{i,kl}$ of $\xi_i\xi_i'$ ($k,l\leq r$), \\
$\frac{1}{N\sqrt{NT}}\sum_{j=1}^N\sum_{i=1}^N\sum_{t=1}^T(u_{it}u_{jt}-Eu_{it}u_{jt})\lambda_{0i}\lambda_{0j}'d_{i,kl}=O_p(1)$,\\
 $\frac{1}{\sqrt{NT}}\sum_{i=1}^N\sum_{t=1}^T(u_{it}^2-Eu_{it}^2)\xi_i\xi_i'=O_p(1).$\\
(iii) For each element $d_{ij, kl}$ of $\xi_i\xi_j'$,\\
$
\frac{1}{N\sqrt{NT}}\sum_{i\neq j, (i,j)\in S_U}\sum_{t=1}^T\sum_{v=1}^N(u_{it}u_{vt}-Eu_{it}u_{vt})\lambda_{0j}\lambda_{0v}'d_{ij, kl}=O_p(1)$,\\
$
\frac{1}{\sqrt{NT}}\sum_{i\neq j, (i,j)\in S_U}\sum_{t=1}^T(u_{it}u_{jt}-Eu_{it}u_{jt})\xi_i\xi_j'=O_p(1)$.

\end{assum}

Under Assumption \ref{ass3.9},  we can achieve the following improved rate of convergence for the averaged  estimation error $\hSigo-\Sigun$:

\begin{thm}\label{l3.1}
Under the assumptions of  Theorem \ref{th3.3} and Assumption \ref{ass3.9},
$$
\frac{1}{N}\|\Lambda_0'[\hSigoi-\Sigun^{-1}]\Lambda_0\|_F=O_p(m_N^2\omega_T^{2-2q}).
$$
\end{thm}

\begin{remark}

\begin{enumerate}

\item 
A simple application of \\
$\|\hSigoi-\Sigun^{-1}\|=O_p(m_N\omega_T^{1-q})$ by Fan et al. (2012) yields \\
$\frac{1}{N}\|\Lambda_0'[\hSigoi-\Sigun^{-1}]\Lambda_0\|_F=O_p(m_N\omega_T^{1-q}).$  In contrast, the rate we present in Theorem \ref{l3.1}  requires more refined asymptotic analysis.    It shows that after weighted by the factor loadings, the averaged convergence rate is faster.
\item
The condition on  the large-entry-set $S_U$ in Assumption \ref{ass3.8add} can be relaxed a bit to $\sum_{i\neq j, (i,j)\in S_U}1=O(N^{1+\epsilon})$ for an arbitrarily small $\epsilon>0$, which will allow  less sparse covariances.   For example, Suppose $\{u_{it}\}_{i\leq N}$ follows a cross sectional  AR($1$) process such that $$u_{it}=\rho u_{i-1,t}+e_{it}$$ for $|\rho|<1$ and $\{e_{it}\}_{i\leq N, t\leq T}$ being independent across both $i$  and $t$. We can then find a partition $S_L\cup S_U$ such that  $\sum_{(i,j)\in S_L}|\Sigma_{u0,ij}|=O(1)$ and $\sum_{i\neq j, (i,j)\in S_U}1=O(N^{1+\epsilon})$ for any $\epsilon>0.$  Theorems \ref{l3.1} and \ref{th3.5} below still hold. But  conditions  in Assumption \ref{ass3.9} need to be adjusted accordingly. For example,  in condition (iii) the normalizing constant $\frac{1}{N\sqrt{N}}$ in the first equation should be changed to $\frac{1}{N^{\epsilon+3/2}}$, and $\frac{1}{\sqrt{N}}$ in the second  equation should be changed to $\frac{1}{N^{(1+\epsilon)/2}}$.  The current Assumption \ref{ass3.9}, on the other hand, keeps our presentation simple.

\end{enumerate}
\end{remark}

\subsubsection{Limiting distribution}

As a result of Theorem \ref{l3.1},   the impact of estimating $\Sigun$ at step one is asymptotically negligible. This  enables us to achieve  the $\sqrt{T}$-consistency and the limiting distribution of $\hlam_j^{(1)}$ for each $j$. We impose further assumptions. 

\begin{assum}\label{ass3.10}
(i)  $\frac{1}{N\sqrt{NT}}\sum_{i=1}^N\sum_{j=1}^N\sum_{t=1}^T(u_{it}u_{jt}-Eu_{it}u_{jt})\xi_i\xi_j'=O_p(1)$.\\
 For each $j\leq N$,$\frac{1}{\sqrt{NT}}\sum_{i=1}^N\sum_{t=1}^T(u_{it}u_{jt}-Eu_{it}u_{jt})\xi_i=O_p(1).$\\
(iii)
$\frac{1}{\sqrt{NT}}\sum_{i=1}^N\sum_{t=1}^T\xi_iu_{it}f_t'=O_p(1)$.

\end{assum}

\begin{thm}\label{th3.5} Suppose $0\leq q<1/2$, and $T=o(N^{2-2q})$. In addition, 
$m_N^2\omega_T^{2-2q}=o(T^{-1/2})$. Then under the assumptions of Theorem \ref{l3.1} and Assumption \ref{ass3.10}, for each $j\leq N$,
$$ \sqrt{T}(\hlam_j^{(1)}-\lamj)\rightarrow^dN_r\left(0,E(u_{jt}f_tf_t')\right).$$
\end{thm}

We make some technical remarks regarding Theorem \ref{th3.5}.
\begin{remark}
\begin{enumerate}
\item The condition $m_N^2\omega_T^{2-2q}=o(T^{-1/2})$ (roughly speaking, this is $m_N=o(T^{1/4})$ when $N$ is very large and $q=0$) strengthens the sparsity condition of Assumption \ref{ass3.1}. The required upper bound for $m_N$ is tight. Roughly speaking,  the estimation error of $\hSigo$ plays a role in the  asymptotic expansion of $\sqrt{T}(\hlam_j^{(1)}-\lamj)$ only through an averaged term as in Theorem \ref{l3.1}. Condition $m_N^2\omega_T^{2-2q}=o(T^{-1/2})$   is required  for that term to be asymptotically negligible. 

\item The asymptotic normality also holds jointly for finitely many estimators. For any finite and fixed $k$, we have, $$
\sqrt{T}(\hlam_1^{(1)'}-\lambda_{01}', \cdots,  \hlam_k^{(1)'}-\lambda_{0k}')'\rightarrow^d N_{kr}(0, E[\cov(u_t^k|f_t)\otimes f_tf_t']).
$$
where  $\cov(u_t^k|f_t)=\cov(u_{1t},...,u_{kt}|f_t)$.
\item   If  Assumption \ref{ass3.10}(i) is replaced by a uniform convergence, by assuming
$
\max_{j\leq N}\|\frac{1}{\sqrt{NT}}\sum_{i=1}^N\sum_{t=1}^T(u_{it}u_{jt}-Eu_{it}u_{jt})\xi_i\|=O_p(\sqrt{N\log N}),
$
 we can then improve the uniform rate of convergence in Theorem \ref{th3.2} and obtain
$$
\max_{j\leq N}\|\hlam_j^{(1)}-\lamj\|=O_p(\sqrt{\frac{\log N}{T}}).
$$
\end{enumerate}
\end{remark}
 
\subsubsection{Estimation of common factors}

For the limiting distribution of $\hfto$,  we make the following additional assumption:
\begin{assum}\label{ass3.11}
There is a positive definite matrix  $Q$ such that  for each $t\leq T,$ 
$$
\frac{1}{N}\Lambda_0'\Sigun^{-1}\Lambda_0\rightarrow Q,\quad \frac{1}{\sqrt{N}}\Lambda_0'\Sigun^{-1}u_t=\frac{1}{\sqrt{N}}\sum_{j=1}^N\xi_j u_{jt}\rightarrow N_r(0, Q).
$$

\end{assum}

For the next assumption, we define $\beta_t=\Sigun^{-1}u_t$. Then  $\beta_t$ has mean zero and covariance matrix  $\Sigun^{-1}$.
\begin{assum}\label{ass3.12}
 For any fixed $t\leq T$,\\
(i) $\frac{1}{\sqrt{NT}}\sum_{s=1}^T\sum_{i=1}^Nf_su_{is}\beta_{it}=O_p(1)$,\\
 $\frac{1}{NT\sqrt{N}}\sum_{i=1}^N\sum_{j=1}^N\sum_{s=1}^T\xi_i(u_{is}u_{js}-Eu_{is}u_{js})\beta_{jt}=o_p(1)$\\
$\frac{1}{T\sqrt{N}}\sum_{i=1}^N\sum_{s=1}^T(u_{is}^2-Eu_{is}^2)\xi_i\beta_{it}=o_p(1)$\\
$
\frac{1}{T\sqrt{N}}\sum_{i\neq j, (i,j)\in S_U}\sum_{s=1}^T(u_{is}u_{js}-Eu_{is}u_{js})\xi_i\beta_{jt}=o_p(1). 
$\\
(ii) For each $k\leq r$,\\
$
\frac{1}{NT\sqrt{N}}\sum_{i=1}^N\sum_{j=1}^N\sum_{s=1}^T(u_{is}u_{js}-Eu_{is}u_{js})\lambda_{0i}\lambda_{0j}'\xi_{ik}\beta_{it}=o_p(1),$\\
$\frac{1}{NT\sqrt{N}}\sum_{i\neq j, (i,j)\in S_U}\sum_{s=1}^T\sum_{l=1}^N(u_{is}u_{ls}-Eu_{is}u_{ls})\lambda_{0j}\lambda_{0l}'\xi_{ik}\beta_{jt}=o_p(1).
$

\end{assum}

\begin{thm}\label{th3.6} Under the assumptions of Theorem \ref{th3.5}, we have for each fixed $t\leq T$,
$$
 \|\hfto-f_t\|=O_p(m_N\omega_T^{1-q}(\log T)^{1/r_1+1/r_2}).
$$where $r_1, r_2>0$ are defined in Assumption \ref{ass3.2}.

If in addition Assumptions \ref{ass3.11}, \ref{ass3.12} are  satisfied and $\sqrt{N}m_N^2\omega_T^{2-2q}=o(1)$. 
Then when  $T^{1/(2-2q)}\ll N\ll T^{2-2q}$,
$$
\sqrt{N}(\hfto-f_t)\rightarrow^d N(0,Q^{-1}).
$$
\end{thm}
\begin{remark}
\begin{enumerate}
\item It follows from Theorem \ref{th3.6}  that for each fixed $t$, $\hfto$ is a root- $N$ consistent estimator of $f_t$.  Root- $N$ consistency for the estimated common factors also holds for the principal components estimator as in Bai (2003). In addition,  the above limiting distribution holds only when $N=o(T^2).$

\item If we strengthen the assumption to $\max_{t\leq T}\|\frac{1}{\sqrt{N}}\sum_{i=1}^N\xi_iu_{it}\|=O_p(\log T)$, then the uniform rate of convergence  can be achieved:
$$
\max_{t\leq T}\|\hfto-f_t\|=O_p(m_N\omega_T^{1-q}(\log T)^{1/r_1+1/r_2+1}).
$$
To compare this rate with that of the PCA estimator, we consider for simplicity, the strictly sparse case $q=0$. Then when $N=o(T^{3/2})$ and $m_N$ is either bounded or growing slowly ($m_N^2\ll\min\{\sqrt{T}, T^{3/2}/N\}$), the above rate is faster than that of the PCA estimator. (The above rate is $O_p((\log T)^{1/r_1+1/r_2}/\sqrt{N})$ when $N=O(T)$, whereas the uniform convergence rate for PCA estimator is $O_p(T^{1/4}/\sqrt{N}).$)

\end{enumerate}
\end{remark}

\section{Joint Estimation}

\subsection{$l_1$- penalized maximum likelihood}

One can also jointly estimate $(\Lambda_0, \Sigun)$  to take into account the cross-sectional dependence and heteroskedasticity simultaneously. As in the sparse covariance estimation literature (e.g., Lam and Fan 2009, Bien and Tibshirani 2011), we penalize the off-diagonal elements of the error covariance estimator, and minimize the following weighted-$l_1$ penalized objective function, motivated by a penalized Gaussian likelihood function:
\begin{eqnarray}\label{eq4.1addd}
(\hLamt, \hSigt)&=&\arg\min_{(\Lambda,\Sigma_u)\in\Theta_{\lambda}\times\Gamma}L_2(\Lambda,\Sigma_u)\cr
&=&\arg\min_{\Lambda\in\Theta_{\lambda}\times\Gamma}\frac{1}{N}\log|\det(\Lambda\Lambda'+\Sigma_u)|+\frac{1}{N}\tr(S_y(\Lambda\Lambda'+\Sigma_u)^{-1})\cr
&&+\frac{1}{N}\sum_{i\neq j} \mu_Tw_{ij}|\Sigma_{u, ij}|,
\end{eqnarray}
where $\Gamma$ is the parameter space for $\Sigma_u$, to be defined later.  We introduce the weighted $l_1$-penalty $N^{-1}\mu_T\sum_{i\neq j} w_{ij}|\Sigma_{u, ij}|$ with $w_{ij}\geq 0$ to penalize the inclusion of many off-diagonal elements  of $\Sigma_{u,ij}$ in small magnitudes, which therefore produces a sparse estimator  $\hSigt$. Here $\mu_T$ is a  tuning parameter that converges to zero at a not-too-fast rate; $w_{ij}$ is an entry-dependent weight parameter, which can be either deterministic or stochastic.   Popular choices of $w_{ij}$ in the literature include: 

\begin{description}
  \item[Lasso] The choice $w_{ij}=1$ for all $i\neq j$ gives the well-known Lasso penalty $N^{-1}\mu_T\sum_{i\neq j}|\Sigma_{u,ij}|$ studied by Tibshirani (1996). The Lasso penalty puts an equal weight to each element of the idiosyncratic covariance matrix.

\item[Adaptive-Lasso] Let $\hSig_{u,ij}^*$ be a preliminary consistent estimator of $\Sigma_{u0, ij}$. Let $
w_{ij}=|\hSig_{u,ij}^*|^{-\gamma}$ for some $\gamma>0$, then 
$$
\frac{\mu_T}{N}\sum_{i\neq j} w_{ij}|\Sigma_{u, ij}|=\frac{\mu_T}{N}\sum_{i\neq j}|\hSig_{u,ij}^*|^{-\gamma}|\Sigma_{u, ij}|
$$
corresponds to the adaptive-lasso penalty proposed by Zou (2006). Note that the adaptive-lasso puts an entry-adaptive weight on each off-diagonal element of $\Sigma_u$, whose reciprocal is proportional to the preliminary estimate. If the true element $\Sigma_{u0, ij}\in S_L$, the weight $|\hSig_{u,ij}^*|^{-\gamma}$ should be quite large, and results in a heavy penalty on that entry. The preliminary estimator $\hSig_{u,ij}^*$ can be taken, for example, as the PCA estimator  $\hSig_{u,ij}^{PCA}=T^{-1}\sum_{t=1}^T\hu_{it}^{PCA}\hu_{jt}^{PCA'}$.  It was shown by Bai (2003) that under mild conditions, $\hSig_{u,ij}^{PCA}-\Sigma_{u0,ij}=O_p(N^{-1/2}+T^{-1/2})$.

\item[SCAD:] Fan and Li (2001) proposed  to use, for some $a>2$ (e.g, $a=3.7$)  $$w_{ij}=
I_{(|\hSig_{u,ij}^*|\leq\mu_T)}+\frac{(a-|\hSig_{u,ij}^*|/\mu_T)_+}{a-1}I_{(|\hSig_{u,ij}^*|>\mu_T)}.
$$
 The notation $z_+$ stands for the positive part  of $z$; $z_+$ is $z$ if $z>0$, zero otherwise. Here $\hSig_{u,ij}^*$  is still a preliminary consistent estimator, which can be taken as the PCA estimator.

\end{description}

\subsection{Consistency of the joint estimation}

We assume the parameter space for $\Sigun$ to be, for some known sufficiently large $M>0$,
$$
\Gamma=\{\Sigma_u: \|\Sigma_u\|_1< M, \|\Sigma_u^{-1}\|_1< M\}.
$$
Then $\Sigun\in\Gamma$ implies that all the eigenvalues of $\Sigun$ are bounded away from both zero and infinity.  There are many examples where both the covariance  and its inverse  have bounded row sums. For example, for each  $t$, when $\{u_{it}\}_{i=1}^N$ follows a cross sectional autoregressive process AR$(p)$ for some fixed $p$, then the maximum row sum of $\Sigma_{u0}$ is bounded. The inverse of $\Sigma_{u0}$ is a banded matrix, whose maximum row sum is also bounded.


As before we assume  $T^{-1}\sum_{t=1}^Tf_tf_t'=I_r$ and $\Lambda_0'\Sigma_{u0}^{-1}\Lambda$ be diagonal for identification. In addition, Assumptions \ref{ass3.2} and \ref{ass3.3} for the two-step estimation are still needed. Those conditions such as strong mixing, weakly dependence and bounded eigenvalues of $N^{-1}\Lambda_0'\Lambda_0$   regulate the data generating process, and asymptotically identify the covariance decomposition (\ref{eq2.2}).

 The conditions for the partition $\{(i,j): i, j \leq N\}=S_L\cup S_U$  of $\Sigun$ are replaced by the following, which are weaker than those of two-step estimation in  Assumption \ref{ass3.8add}. Define the number of off-diagonal large entries:
\begin{equation}
D=\sum_{i\neq j, (i,j)\in S_U}1.
\end{equation}

\begin{assum} \label{ass4.1}There exists a partition $\{(i,j): i\leq N, j\leq N\}=S_L\cup S_U$ where $S_U$ and $S_L$ are disjoint, which satisfies:\\
(i) $\Sigma_{u0, ii}\in S_U$ for all $i\leq N$,\\
(ii) $D=o(\min\{N\sqrt{T/\log N}, N^2/\log N\}),$\\
(iii) $\sum_{(i,j)\in S_L}|\Sigma_{u0, ij}|=o(N).$
\end{assum}

The following assumption is imposed on the penalty parameters. Define the weights ratios
$$
\alpha_T=\frac{\max_{i\neq j, (i,j)\in S_U}w_{ij}}{\min_{(i,j)\in S_L}w_{ij}}, \hspace{1em}\beta_T=\frac{\max_{(i,j)\in S_L}w_{ij}}{\min_{(i,j)\in S_L}w_{ij}}.
$$

\begin{assum} \label{ass4.2}The tuning parameter $\mu_T$ and the weights $\{w_{ij}\}_{i\leq N, j\leq N}$ satisfy:\\
(i)
$$
\alpha_T=o_p\left[\min\left\{
\sqrt{\frac{T}{\log N}}\frac{N}{D},\left(\frac{T}{\log N}\right)^{1/4}\sqrt{\frac{N}{D}}, \frac{N}{\sqrt{D
\log N}}\right\}\right],
$$
$$
\beta_T\sum_{(i,j)\in S_L}|\Sigma_{u0, ij}|=o_p(N),
$$
(ii)
$\mu_T\max_{(i,j)\in S_L}w_{ij}\sum_{(i,j)\in S_L}|\Sigma_{u0, ij}|=o(\min\{N, N^2/D, N^2/(D\alpha_T^2)\}),$\\
$\mu_T\max_{i\neq j, (i,j)\in S_U}w_{ij}=o(\min\{N/D, \sqrt{N/D}, N/(D\alpha_T)\}),$\\
$\mu_T\min_{(i,j)\in S_L}w_{ij}\gg\sqrt{\log N/T}+(\log N)/N.$

\end{assum}
The above assumption is not as complicated as it looks, and is satisfied by many examples. For instance, the Lasso penalty sets $w_{ij}=1$ for all $i, j\leq N$. Hence $\alpha_T=\beta_T=1.$ Then condition (i) of Assumption \ref{ass4.2} follows from Assumption \ref{ass4.1}(ii), which is also satisfied if $D=O(N)$. Condition (ii) is also straightforward to verify. This immediately implies the following lemma.
\begin{lem}[Lasso] \label{l4.1} Choose $w_{ij}=1$ for all $i,j\leq N, i\neq j$. Suppose in addition $D=O(N)$ and $\log N=o(T)$. Then Assumption \ref{ass4.2} is satisfied if the tuning parameter $\mu_T=o(1)$ is such that
$$
\sqrt{\frac{\log N}{T}}+\frac{\log N}{N}=o(\mu_T).
$$
\end{lem}
One of the attractive features of this lemma is that the condition on $\mu_T$ does not depend on the   unknown $\Sigma_{u0}.$ We will present the adaptive lasso and SCAD as another two examples of the weighted-$l_1$ penalty in Section 4.3 below, both satisfy the above assumption.


Our main theorem is stated as follows.

\begin{thm}\label{th4.1} Suppose $\log N=o(T)$. Under Assumptions \ref{ass3.2}, \ref{ass3.3},  \ref{ass3.7}, \ref{ass4.1}, and \ref{ass4.2}, the penalized ML estimator satisfies: as $T$ and $N\rightarrow\infty,$
$$
\frac{1}{N}\|\hSigt-\Sigma_{u0}\|^2_F\rightarrow^p0, \quad
\frac{1}{N}\|\hLamt-\Lambda_{0}\|^2_F\rightarrow^p0.
$$  For each $t\leq T$,
$$
\|\hftt-f_t\|=o_p(1).
$$
\end{thm}
\begin{remark}
\begin{enumerate}
\item
The consistency for $\hftt$ can be made uniformly in $t\leq T$ if the condition is strengthened to $\max_{t\leq T}\|N^{-1/2}\sum_{i=1}^N\xi_iu_{it}\|=o_p(\sqrt{N})$.
\item
To establish the consistency in the high dimensional literature, one usually    constructs a neighborhood of the true parameters $(\Lambda_0,\Sigma_{u0})\in U$ 
(e.g., Rothman et al. 2008, Lam and Fan 2009), and show that with probability approaching one, $ L_2(\Lambda_0,\Sigma_{u0})>\sup_{(\Lambda,\Sigma_u)\notin U}L_2(\Lambda,\Sigma_u)$.  This strategy however, does not work  here due to the technical difficulty in dealing with the term $(\Lambda\Lambda'+\Sigma_u)$ in the likelihood function, because its largest $r$  eigenvalues are unbounded and grow at rate $O(N)$ uniformly in the parameter space.  One of the contributions of Theorem \ref{th4.1} is to achieve  consistency using  a new strategy to deal with the penalized  likelihood function, which  involves diverging eigenvalues.

 \end{enumerate}

\end{remark}
 
 In this paper we only present the consistency for the joint estimation, which  is already technically difficult as one needs to deal with an equilibrium of the first order conditions for both $(\hLamt,\hSigt)$ simultaneously. Deriving the limiting distributions for the joint estimators is difficult, and we leave this as a future topic.


\subsection{Two examples}

We present two popular choices for the weights as examples: one is adaptive lasso, proposed by Zou (2006), and the other is SCAD by Fan and Li (2001).  Both weights depend on a preliminary consistent estimate of each element of   $\Sigma_{u0}$. In the high dimensional approximate factor model, a simple consistent estimate for each element can be obtained by the principal component analysis (Stock and Watson 1998 and Bai 2003).  

To simplify the presentation, we  will assume that $D=O(N)$, which controls the number of off-diagonal large entries of $\Sigun$. Moreover, we retain Assumption \ref{ass3.6}: 
$$
\max\{|\Sigma_{u0, ij}|: (i,j)\in S_L\}\ll\omega_T\ll \min\{|\Sigma_{u0, ij}|: \Sigma_{u0, ij}\in S_U\}, $$
and recall that $\omega_T=\sqrt{\frac{\log N}{T}}+\frac{1}{\sqrt{N}}$.

 Let the initial estimate $\hSig_{u,ij}^*=R_{ij}$, where $R_{ij}$ is the PCA estimator of $\Sigma_{u0,ij}$ as in Bai (2003). The adaptive lasso chooses the weights to be, for some constant $\gamma\in(0, 1]$,
\begin{equation}\label{eq4.3}
(\text{Adaptive Lasso}):\quad w_{ij}=(|\hSig_{u,ij}^*| +\delta_T)^{-\gamma},
\end{equation}
where $\delta_T=o(1)$ is a pre-determined nonnegative sequence. The additive $\delta_T$ was not included in the original definition of adaptive lasso in Zou (2006), but has often been seen in recent literature, e.g., Xue and Zou (2012).  We include it here in the weights to prevent $w_{ij}$ getting too large if $|\hSig_{u,ij}^*|$ is very close to zero. The adaptive lasso has been used extensively in the high dimensional literature, see for example, Huang, Ma and Zhang (2006), van de Geer, B\"{u}hlmann and Zhou (2011), Caner and Fan (2011), etc.



 Another important example is SCAD, defined as: for some $a>2$,
 \begin{equation}\label{eq4.4}
 (\text{SCAD}):   \quad w_{ij}=
I_{(|\hSig_{u,ij}^*|\leq\mu_T)}+\frac{(a-|\hSig_{u,ij}^*|/\mu_T)_+}{a-1}I_{(|\hSig_{u,ij}^*|>\mu_T)}.
 \end{equation}

We have the following theorem.
\begin{thm}\label{th4.2} Suppose either the Adaptive Lasso or SCAD is used for the weighted-$l_1$ penalized objective function.  Also, suppose $ \log N=o(T)$, $D=O(N)$, $\sum_{(i,j)\in S_L}|\Sigma_{u0,ij}|=o(N)$ and Assumptions \ref{ass3.2}, \ref{ass3.3},  \ref{ass3.7}, \ref{ass4.1}  hold.
 In addition, assume  the tuning parameters are such that:\\
 (i) for  Adaptive Lasso,  
 \begin{equation}\label{e4.2}
\omega_T\left(\frac{\sum_{(i,j)\in S_L}|\Sigma_{u0, ij}|}{N}\right)^{1/\gamma}\ll\delta_T\ll \omega_T,
\end{equation}
\begin{equation}\label{e4.3}
\omega_T^{1+\gamma}\ll\mu_T\ll\omega_T^{\gamma};
\end{equation}
(ii) for SCAD:
 \begin{equation}\label{eq4.7}
 \left(\frac{\log N}{T}\right)^{1/4}+\left(\frac{\log N}{N}\right)^{1/2}\ll\mu_T\ll \min_{i\neq j, (i,j)\in S_U}|\Sigma_{u0,ij}|.
 \end{equation}
 
  Then Assumption \ref{ass4.2} is satisfied, and
 $$
\frac{1}{N}\|\hSigt-\Sigma_{u0}\|^2_F=o_p(1), \quad
\frac{1}{N}\|\hLamt-\Lambda_{0}\|^2_F=o_p(1).
$$
$$
\|\hftt-f_t\|=o_p(1).
$$
\end{thm}

As in the case of Lemma \ref{l4.1},  an  attractive feature of this theorem is that, if both the upper bound of $\sum_{(i,j)\in S_L}|\Sigma_{u0, ij}|$ and the lower bound of\\ $\min_{i\neq j, (i,j)\in S_U}|\Sigma_{u0,ij}|$  are known,  [e.g., in the strictly sparse model, \\$\sum_{(i,j)\in S_L}|\Sigma_{u0, ij}|=0$, and assume  $\min_{i\neq j, (i,j)\in S_U}|\Sigma_{u0,ij}|$ is bounded away from zero as in MA(1)] then Conditions (\ref{e4.2}) - (\ref{eq4.7}) do not depend on any other unknown feature of $\Sigma_{u0}$.

\section{Numerical Examples}

We propose a novel algorithm to numerically minimize the  objective function $L_2(\Lambda, \Sigma_u)$ (\ref{eq4.1addd}) for joint estimation, which combines the EM algorithm with the majorize-minimize method recently proposed by Bien and Tibshirani (2011). The algorithm uses the PCA as initial values, and updates the estimator iteratively. At each iteration, an EM-algorithm is carried out to estimate $\Lambda$ and the empirical residual covariance $\frac{1}{T}\sum_{t=1}^T\hu_t\hu_t'.$ Then a majorize-minimize method (Bien and Tibshirani 2011) is used to obtain a positive definite estimate of the covariance $\Sigma_u$ based on $\frac{1}{T}\sum_{t=1}^T\hu_t\hu_t'$ and soft-thresholding. The algorithm is summarized as follows (see Bai and Li (2012) and Bien and Tibshirani (2011) for detailed descriptions of the algorithm).

\begin{enumerate}
  \item Initialize $\hLam$ and $\widehat{u}$ as the PCA estimators. Initialize $\hSig_u$ as a diagonal matrix of  the sample covariance based on the PCA residuals.
\item At step k+1,     $\hLam_{k+1}=AM^{-1}$, where \\$
M=\hLam_k'\hSig_{y,k}^{-1}S_y\hSig_{y,k}^{-1}\hLam_k+I_r-\hLam_k'\hSig_{y,k}^{-1}\hLam_k
$,
$$
A=S_y\hSig_{y,k}^{-1}\hLam_k,\quad
\hSig_{y,k}=\hLam_k\hLam_k'+\hSig_{u,k}.
$$
 Let $
S_{u,k}=S_y-A\hLam_{k+1}'-\hLam_{k+1}A'+\hLam_{k+1}M\hLam_{k+1}'.
$
\item Still at step $k+1$, For some small value $t>0$ , let $B=\hSig_{u,k}-t(\hSig_{u,k}^{-1}-\hSig_{u,k}^{-1} S_{u,k}\hSig_{u,k}^{-1})$.  Let
$$
\hSig_{u,k+1}=S(B, \lambda tK)
$$
where  $S(A,B)_{ij}=sign(A_{ij})(A_{ij}-B_{ij})$ and $K$ is a matrix  whose off-diagonal $K_{ij}$ is $|(S_{u,k})_{ij}|^{-\gamma}$ and diagonal elements are zero.  

\item Repeat  2-3 until converge.
\end{enumerate}

We present a numerical experiment to illustrate the performance of the proposed method. The data was generated as following:  $\{e_{it}\}_{i\leq N, t\leq T}$ are both serially and cross-sectionally independent as $N(0,1)$. Let
$$u_{1t}=e_{1t}, \hspace{1em}u_{2t}=e_{2t}+a_1e_{1t},  \hspace{1em}u_{3t}=e_{3t}+a_2e_{2t}+b_1e_{1t},$$
$$
u_{i+1,t}=e_{i+1,t}+a_ie_{it}+b_{i-1}e_{i-1,t}+c_{i-2}e_{i-2,t},
$$
where $\{a_i, b_i, c_i\}_{i=1}^N$ are  i.i.d. $N(0,0.7^2)$.  Let the two factors $\{f_{1t}, f_{2t}\}$ be i.i.d. $N(0,1)$, and $\{\lambda_{i,1}, \lambda_{i,2}\}_{i\leq N}$ be uniform on $[0,1]$. Then $\Sigma_{u0}$ is a banded matrix.   

We apply the adaptive lasso penalty for our joint estimation, with various choices of the tuning parameters $\gamma$ and $\mu_T$. The result is compared with the PCA estimator and the regular maximum likelihood restricted to diagonal $\hSig_u$ (DML, Bai and Li 2012).  More specifically, DML estimates $(\Lambda_0,\Sigun)$ by:
\begin{equation}\label{eq5.1}
\min_{\Sigma_{u,ij}=0 \text{  for }i\neq j}\min_{\Lambda}\frac{1}{N}\log|\Lambda\Lambda'+\Sigma_u|+\frac{1}{N}\tr(S_y(\Lambda\Lambda'+\Sigma_u)^{-1}).
\end{equation}
Therefore DML forces the covariance estimator to be diagonal even though the true $\Sigun$ is not. Hence it does not take the idiosyncratic cross-sectional dependence into account.

For each estimator, the smallest canonical correlation (the higher the better) between the estimator and the parameter has been used as a measurement to assess the accuracy of each estimator.  Tables  \ref{table1} and \ref{table2} list the results of the estimated factor loadings and common factors from joint-estimation.



\begin{table}[htdp]
\caption{Canonical correlations between $\hLamt$ and $\Lambda_0$}
\begin{center}
\begin{tabular}{cc|cc|cc|cccc}
 
\hline
&& \multirow{3}{*}{PCA} & \multirow{3}{*}{DML} & \multicolumn{4}{c}{Penalized ML}     \\
&&  && \multicolumn{2}{c}{$\gamma=1$} &   \multicolumn{2}{c}{$\gamma=5$} \\
$T$ & $N$ && & $\mu_T=0.08$ & $\mu_T=0.3$ & $ \mu_T=0.08$ & $\mu_T=0.3$\\
\hline
50 & 50 & 0.205 &  0.199&  0.212  &  0.222  &   0.230& 0.234 \\
50 & 100 &  0.429 &  0.558 &  0.591&  0.613 &   0.627  &  0.631    \\
50 & 150 &   0.328 &  0.470 &  0.494 & 0.495  & 0.515 &  0.507  \\
& & & & & & & & \\
100 & 50 &  0.496&  0.519 & 0.560 &    0.537&  0.558 & 0.537   \\
100 & 100 &  0.394  &  0.574 &  0.621  &   0.648 &   0.648 &  0.658 \\
100 & 150 &   0.774  & 0.819& 0.837  &   0.829& 0.840&  0.836   \\
\hline
\end{tabular}
\label{table1}
\small

\it Canonical correlations are presented.  DML is defined in (\ref{eq5.1}) which treats $\Sigma_u$ to be diagonal. Penalized ML uses the one-step adaptive Lasso estimation.
\end{center}
\end{table}

\begin{table}[htdp]
\caption{Canonical correlations between $\widehat{F}^{(2)}$ and $F$}
\begin{center}
\begin{tabular}{cc|cc|cc|cccc}
 
\hline
&& \multirow{3}{*}{PCA} & \multirow{3}{*}{DML} & \multicolumn{4}{c}{Penalized ML}     \\
&&  && \multicolumn{2}{c}{$\gamma=1$} &   \multicolumn{2}{c}{$\gamma=5$} \\
$T$ & $N$ && & $\mu_T=0.08$ & $\mu_T=0.3$ & $ \mu_T=0.08$ & $\mu_T=0.3$\\
\hline
50 & 50 & 0.232 &  0.234  &  0.251 &  0.267 & 0.279  & 0.283  \\
50 & 100 &   0.477& 0.640 & 0.671 &   0.732 & 0.748 &  0.749    \\
50 & 150 & 0.411&   0.599 &  0.623 & 0.638  & 0.666 &  0.650  \\
& & & & & & & & \\
100 & 50 & 0.430  &  0.446 & 0.503& 0.473 & 0.508  & 0.474   \\
100 & 100 &  0.371 &  0.579 &  0.647 &  0.688 &   0.687 &  0.697   \\
100 & 150 &  0.820  & 0.867&   0.880&   0.892 &  0.912&    0.903 \\
\hline
\end{tabular}
\label{table2}
\small

\it Canonical correlations are presented. Penalized ML uses the one-step adaptive Lasso estimation.
\end{center}
\end{table}

We have also computed the canonical correlations between the estimators and the true parameters using the regularized two-step method (Section 3) with iterations. For computational simplicity, the threshold value in the first step has been fixed to be the adaptive threshold of Fan et al. (2012) with a universal constant $C=1$, which we find to maintain the finite-sample positive definiteness well. The results demonstrate that both two-step and joint estimations have  higher canonical correlations, and thus outperform the PCA and DML.   

Our  EM plus majorize-minimize algorithm maximizes an approximate penalized likelihood function. Developing an efficient algorithm for maximizing the original  likelihood function will be a future research direction.



\begin{table}[htdp]
\caption{Canonical correlations between the regularized two-step ML estimators (Section 3) and the true parameters}
\begin{center}
\begin{tabular}{cc|ccc|ccc}
 
\hline
&& \multicolumn{3}{c|}{ Factor loadings}  &\multicolumn{3}{c}{ Factors}   \\
$T$ & $N$ &PCA& DML &Two-step &PCA&DML &Two-step\\
 & & &  & ML & & & ML\\
\hline
50 & 50 & 0.205 &  0.199&  0.241 &  0.232 &  0.234    & 0.277\\
50 & 100 &  0.429 &  0.558 &       0.643&  0.477& 0.640&    0.752   \\
50 & 150 &   0.328 &  0.470 & 0.565 & 0.411&   0.599  &  0.731 \\
& & & & & &  \\
100 & 50 &  0.496&  0.519 &  0.548  &     0.430  &  0.446  &0.469 \\
100 & 100 &  0.394  &  0.574 &   0.717  & 0.371 &  0.579 &0.758 \\
100 & 150 &   0.774  & 0.819&   0.846  &   0.820  & 0.867 &    0.927  \\
\hline
\end{tabular}
\label{table3}
\small

\it The SCAD$(\tau_{ij})$ threshold has been used for the covariance estimation, where $\tau_{ij}=\alpha_{ij}\omega_T$ with the adaptive threshold constant $\alpha_{ij}$ proposed by Cai and Liu (2011).
\end{center}
\end{table}

\section{Conclusion}
We study the estimation of a high dimensional approximate factor model in the presence of cross sectional dependence and heteroskedasticity. The classical PCA method does not efficiently estimate the factor loadings or common factors because it essentially treats the idiosyncratic error to be homoskedastic and cross sectionally uncorrelated.  For the efficient estimation it is   essential to estimate a large error covariance matrix.

We assume the model to be conditionally sparse in the sense that after the common factors are taken out, the idiosyncratic components have a sparse covariance matrix. This enables us to combine the merits of both sparsity  and high dimensional factor analysis. Two maximum-likelihood-based approaches are proposed to estimate the common factors and factor loadings, both involve regularizing a large covariance sparse matrix. Extensive asymptotic analysis has been carried out. In particular, we develop the inferential theory for the two-step estimation.

It remains to derive the limiting distribution as well as the optimal rates of convergence for the estimators by the joint-estimation method. This will extend the consistency results obtained in the current paper.  In the presence of a covariance $\Lambda_0\Lambda_0'$ that has fast-diverging eigenvalues, the task is difficult because it requires  the consistency of the penalized covariance estimator under the operator norm.    We intend to address this issue  in  future research.

\appendix

\section{Proofs for generic estimators}

We need to establish the results for two sets of estimators: the two-step estimator and the joint estimator, whose proofs for consistency share some similarities. Therefore in this section we establish some preliminary results for generic estimators that can be used for both cases. We denote by $(\hLam,\hSig_u)$ as a generic estimator for $(\Lambda_0, \Sigun)$, which can be either $(\hLamo,\hSigo)$ or $(\hLam^{(2)}, \hSig_u^{(2)})$. Define
\begin{equation}\label{q2}
Q_2(\Lambda, \Sigma_u)=\frac{1}{N}\tr(\Lambda_0'\Sigma_u^{-1}\Lambda_0-\Lambda_0'\Sigma_u^{-1}\Lambda(\Lambda'\Sigma_u^{-1}\Lambda)^{-1}
\Lambda'\Sigma_u^{-1}\Lambda_0),
\end{equation}
\begin{equation}\label{q3}
Q_3(\Lambda,\Sigma_u)=\frac{1}{N}\log|\Lambda\Lambda'+\Sigma_u|+\frac{1}{N}\tr(S_y(\Lambda\Lambda'+\Sigma_u)^{-1})-\frac{1}{N}\tr(S_u\Sigma_u^{-1})-\frac{1}{N}\log|\Sigma_u|-Q_2( \Lambda, \Sigma_u).
\end{equation}
Define the set
\begin{eqnarray*}
\Xi_{\delta}=\{(\Lambda,\Sigma_u): &&\delta^{-1}<\lambda_{\min}(N^{-1}\Lambda'\Lambda)\leq\lambda_{\max}(N^{-1}\Lambda'\Lambda)<\delta, 
\cr
&&\delta^{-1}<\lambda_{\min}(\Sigma_u)\leq\lambda_{\max}(\Sigma_u)<\delta\}
\end{eqnarray*}

We first present a   lemma that will be needed throughout the proof.
\begin{lem}\label{la.1}  
(i) $\max_{i,j\leq r}|\frac{1}{T}\sum_{t=1}^Tf_{it}f_{jt}-Ef_{it}f_{jt}|=O_p(\sqrt{1/T})$.\\
(ii) $\max_{i,j\leq N}|\frac{1}{T}\sum_{t=1}^Tu_{it}u_{jt}-Eu_{it}u_{jt}|=O_p(\sqrt{(\log N)/T})$.\\
(iii) $\max_{i\leq r, j\leq N}|\frac{1}{T}\sum_{t=1}^Tf_{it}u_{jt}|=O_p(\sqrt{(\log N)/T})$.
\end{lem}
\begin{proof}  See Lemmas A.3 and B.1 in Fan, Liao and Mincheva (2011). \end{proof}

\begin{lem} \label{la.2}  Under Assumption  3.2, for any $\delta>0$,  $$\sup_{(\Lambda,\Sigma_u)\in\Xi_{\delta}}|Q_3(\Lambda,\Sigma_u)|=O\left(\frac{\log N}{N}+\sqrt{\frac{\log N}{T}}\right).$$
Therefore we can write
\begin{eqnarray}
&&\frac{1}{N}\log|\Lambda\Lambda'+\Sigma_u|+\frac{1}{N}\tr(S_y(\Lambda\Lambda'+\Sigma_u)^{-1})\cr
&&=\frac{1}{N}\tr(S_u\Sigma_u^{-1})+\frac{1}{N}\log|\Sigma_u|+Q_2( \Lambda, \Sigma_u)+O\left(\frac{\log N}{N}+\sqrt{\frac{\log N}{T}}\right).
\end{eqnarray}
\end{lem}

\begin{proof}

First of all, note that $|\Lambda\Lambda'+\Sigma_u|=|\Sigma_u|\times|I_r+\Lambda'\Sigma_u^{-1}\Lambda|$,  and  $\sup_{(\Lambda,\Sigma_u)\in\Xi_{\delta}}\frac{1}{N}\log|I_r+\Lambda'\Sigma_u^{-1}\Lambda|=O\left(\frac{\log N}{N}\right),$ hence we have
\begin{equation}\label{eqa.2}
\frac{1}{N}\log|\Lambda\Lambda'+\Sigma_u|=\frac{1}{N}\log|\Sigma_u|+O\left(\frac{\log N}{N}\right),
\end{equation}
where $O(\cdot)$ is uniform in $\Xi_{\delta}$.  Equation (\ref{eqa.2}) will be used later in the proof.

We now consider the term $N^{-1}\tr(S_y(\Lambda\Lambda'+\Sigma_u)^{-1})$.   With  the identification condition $\frac{1}{T}\sum_{t=1}^Tf_tf_t'=I_r,$ $\bar{f}=0,$ and  $S_u=\frac{1}{T}\sum_{t=1}^Tu_tu_t'$, 
  $$S_y=\frac{1}{T}\sum_{t=1}^T(y_t-\bar{y})(y_t-\bar{y})'=\Lambda_0\Lambda_0'+S_u+\Lambda_0\frac{1}{T}\sum_{t=1}^T f_t u_t'+(\Lambda_0\frac{1}{T}\sum_{t=1}^T f_t u_t')'-\bar{u}\bar{u}'.$$
By the matrix inversion formula  $(\Lambda\Lambda'+\Sigma_u)^{-1}=\Sigma_u^{-1}-\Sigma_u^{-1}\Lambda(I_r+\Lambda'\Sigma_u^{-1}\Lambda)
^{-1}\Lambda'\Sigma_u^{-1}$, 
\begin{equation}\label{eqa.3}
\frac{1}{N}\tr(S_y(\Lambda\Lambda'+\Sigma_u)^{-1})=\frac{1}{N}\tr(\Lambda_0'\Sigma_u^{-1}\Lambda_0)+\frac{1}{N}\tr(S_u\Sigma_u^{-1})-A_1+A_2+A_3-A_4-A_5,
\end{equation}
where
$A_1=N^{-1}\tr(\Lambda_0\Lambda_0'\Sigma_u^{-1}\Lambda(I_r+\Lambda'\Sigma_u^{-1}\Lambda)
^{-1}\Lambda'\Sigma_u^{-1}),$
$
A_2=\frac{1}{N}\tr(\frac{1}{T}\sum_{t=1}^T\Lambda_0 f_t u_t'(\Lambda\Lambda'+\Sigma_u)^{-1})
$, $
A_3=\frac{1}{N}\tr(\frac{1}{T}\sum_{t=1}^Tu_tf_t'\Lambda_0'(\Lambda\Lambda'+\Sigma_u)^{-1}),
$ and $
A_4=\frac{1}{N}\tr(S_u\Sigma_u^{-1}\Lambda(I_r+\Lambda'\Sigma_u^{-1}\Lambda)
^{-1}\Lambda'\Sigma_u^{-1}).$ Term $A_5=N^{-1}\tr(\bar{u}\bar{u}'(\Lambda\Lambda'+\Sigma_u)^{-1})=O_p((\log N)/T)$ uniformly in the parameter space, and hence can be ignored.

Let us look at terms $A_1, A_2, A_3$ and $A_4$ subsequently.  Note that $\lambda_{\max}(\Sigma_u)$ and $N\lambda^{-1}_{\min}(\Lambda'\Lambda)$ are both bounded from above uniformly in $\Xi_{\delta}$, we have,
\begin{equation}\label{eqa.4}
\sup_{(\Lambda,\Sigma_u)\in\Xi_{\delta}}\lambda_{\max}[(\Lambda'\Sigma_u^{-1}\Lambda)^{-1}]\leq\sup_{(\Lambda,\Sigma_u)\in\Xi_{\delta}}\frac{\lambda_{\max}(\Sigma_u)}{\lambda_{\min}(\Lambda'\Lambda)}=O(N^{-1}),
\end{equation}
\begin{equation}\label{eqa.5}
\sup_{(\Lambda,\Sigma_u)\in\Xi_{\delta}}\lambda_{\max}[(I_r+\Lambda'\Sigma_u^{-1}\Lambda)^{-1}]\leq\sup_{(\Lambda,\Sigma_u)\in\Xi_{\delta}}\lambda_{\max}[(\Lambda'\Sigma_u^{-1}\Lambda)^{-1}]=O(N^{-1}).
\end{equation}
In addition, $\|\Lambda\|_F=O(\sqrt{N})$, $\lambda_{\max}(\Sigma_u^{-1})=O(1)$ uniformly in $\Xi_{\delta}$,  and $ \|\Lambda_0\|_F=O(\sqrt{N})$. Applying the matrix inversion formula   yields
\begin{eqnarray}\label{eqa.6}
A_1&=&\frac{1}{N}\tr(\Lambda_0'\Sigma_u^{-1}\Lambda(\Lambda'\Sigma_u^{-1}\Lambda)
^{-1}\Lambda'\Sigma_u^{-1}\Lambda_0)-\frac{1}{N}\tr(\Lambda_0'\Sigma_u^{-1}\Lambda(\Lambda'\Sigma_u^{-1}\Lambda)^{-1}(I_r+\Lambda'\Sigma_u^{-1}\Lambda)
^{-1}\Lambda'\Sigma_u^{-1}\Lambda_0)\cr
&=&\frac{1}{N}\tr(\Lambda_0'\Sigma_u^{-1}\Lambda(\Lambda'\Sigma_u^{-1}\Lambda)
^{-1}\Lambda'\Sigma_u^{-1}\Lambda_0)+O\left(\frac{1}{N}\right),
\end{eqnarray}
where $O(\cdot)$ is uniform over $(\Lambda, \Sigma_u)\in\Xi_{\delta}$.  In the second equality above we applied (\ref{eqa.4}) and (\ref{eqa.5}) and  the following inequality:
\begin{eqnarray*}
&&\frac{1}{N}\tr(\Lambda_0'\Sigma_u^{-1}\Lambda(\Lambda'\Sigma_u^{-1}\Lambda)^{-1}(I_r+\Lambda'\Sigma_u^{-1}\Lambda)
^{-1}\Lambda'\Sigma_u^{-1}\Lambda_0)\cr
&&\leq \frac{1}{N}\|\Lambda_0'\Sigma_u^{-1}\Lambda\|_F^2\lambda_{\max}[(\Lambda'\Sigma_u^{-1}\Lambda)^{-1}]
\lambda_{\max}[(I_r+\Lambda'\Sigma_u^{-1}\Lambda)^{-1}]\cr
&&\leq O(N^{-3})\|\Lambda_0\|_F^2\|\Lambda\|_F^2\lambda_{\max}(\Sigma_u^{-1})=O(N^{-1}).
\end{eqnarray*}
By Lemma \ref{la.1}(iii), and $\lambda_{\max}((\Lambda\Lambda'+\Sigma_u)^{-1})\leq\lambda_{\max}(\Sigma_u^{-1})=O(1)$ uniformly in $\Xi_{\delta}$, 
\begin{eqnarray}\label{eqa.7}
\sup_{(\Lambda,\Sigma_u)\in\Xi_{\delta}}  |A_2|&\leq&\frac{1}{N}\|\Lambda_0'(\Lambda\Lambda'+\Sigma_u)^{-1}\|_F\left\|\frac{1}{T}\sum_{t=1}^T f_t u_t'\right\|_F=O_p(\sqrt{\frac{\log N}{T}}).
\end{eqnarray} Similarly,  $\sup_{(\Lambda,\Sigma_u)\in\Xi_{\delta}}  |A_3|=O_p(\sqrt{\frac{\log N}{T}}).
$ Again by the matrix inversion formula,
$$A_4=\frac{1}{N}\tr(S_u\Sigma_u^{-1}\Lambda(\Lambda'\Sigma_u^{-1}\Lambda)
^{-1}\Lambda'\Sigma_u^{-1})
-\frac{1}{N}\tr(S_u\Sigma_u^{-1}\Lambda(\Lambda'\Sigma_u^{-1}\Lambda)
^{-1}(I+\Lambda'\Sigma_u^{-1}\Lambda)^{-1}\Lambda'\Sigma_u^{-1}).
$$
The second term on the right hand side is of smaller order (uniformly) than the first term,  because it has an additional term $(I+\Lambda'\Sigma_u^{-1}\Lambda)^{-1}$, whose maximum eigenvalue is $O(N^{-1})$ uniformly by (\ref{eqa.5}). The first term is bounded by (uniformly in $\Xi_{\delta}$ ):
$$
 \frac{c}{N}\| S_u\Sigma_u^{-1}\Lambda\|_F O(N^{-1})\|\Lambda'\Sigma_u^{-1}\|_F\leq O(N^{-1})\lambda_{\max}(S_u)=O(\sqrt{\frac{\log N}{T}}+\frac{1}{N}). 
$$
Hence $\sup_{(\Lambda,\Sigma_u)\in\Xi_{\delta}} |A_4|=O(T^{-1/2}(\log N)^{1/2}+N^{-1}). $
Results (\ref{eqa.2}) and (\ref{eqa.3}) then  yield
\begin{eqnarray*}
&&\frac{1}{N}\log|\Lambda\Lambda'+\Sigma_u|+\frac{1}{N}\tr(S_y(\Lambda\Lambda'+\Sigma_u)^{-1})\cr
&&=\frac{1}{N}\tr(\Lambda_0'\Sigma_u^{-1}\Lambda_0)+\frac{1}{N}\tr(S_u\Sigma_u^{-1})+\frac{1}{N}\log|\Sigma_u|-\frac{1}{N}\tr(\Lambda_0'\Sigma_u^{-1}\Lambda(\Lambda'\Sigma_u^{-1}\Lambda)
^{-1}\Lambda'\Sigma_u^{-1}\Lambda_0)\cr
&&+O\left(\frac{\log N}{N}+\sqrt{\frac{\log N}{T}}\right)\cr
&&=\frac{1}{N}\tr(S_u\Sigma_u^{-1})+\frac{1}{N}\log|\Sigma_u|+Q_2( \Lambda, \Sigma_u)+O\left(\frac{\log N}{N}+\sqrt{\frac{\log N}{T}}\right).
\end{eqnarray*}

\end{proof}

Throughout the proofs, we note that the consistency depends crucially on the consistency of the following quantities:
$$
J=(\hLam-\Lambda_0)'\hSig_u^{-1}\hLam(\hLam'\hSig_u^{-1}\hLam)^{-1}
$$
We state the following lemma for the generic estimators. 
\begin{lem} \label{assa.1}
(i) $\Lambda_0' \Sigma_{u0}^{-1}\Lambda_0 -(I_r- J )\hLam'\hSig_u^{-1}\hLam(I_r- J )'=o_p(N)$\\
(ii) First order condition: $\hLam'(\hLam\hLam'+\hSig_u)^{-1}(S_y-\hLam\hLam'-\hSig_u)=0.$
\end{lem}
We will prove Lemma \ref{assa.1} for  both   $(\hLamo,\hSigo)$ and $(\hLam^{(2)}, \hSig_u^{(2)})$ later when we deal with these two estimators individually.

\begin{lem}\label{la.3}
Suppose Lemma \ref{assa.1} holds, then \\
(i) $\hLam'\hSig_u^{-1}(S_y-\hLam\hLam'-\hSig_u)=0.$\\
(ii) $( J -I_r)'( J -I_r)-I_r=O_p(N^{-1}+T^{-1/2}(\log N)^{1/2}).$
\end{lem}

\begin{proof}

(i) Using the matrix inverse formula, the same argument of Bai and Li (2012)'s (A.2) implies $\hLam'(\hLam\hLam'+\hSig_u)^{-1}=(I_r+\hLam'\hSig_u^{-1}\hLam)^{-1}\hLam'\hSig_u^{-1}$. Thus part (i) follows from the first order condition in Lemma \ref{assa.1}.

(ii) Let $H=(\hLam'\hSig_u^{-1}\hLam)^{-1}$. Part (i) can be equivalently written as 
$J + J'-J'J +K=0$
where
$$
K= J ' \frac{1}{T}\sum_{t=1}^T f_t u_t' \hSig_u^{-1}\hLam H + H \hLam'\hSig_u^{-1} \frac{1}{T}\sum_{t=1}^Tu_tf_t' J   - \frac{1}{T}\sum_{t=1}^T f_t u_t' \hSig_u^{-1}\hLam H - H \hLam'\hSig_u^{-1} \frac{1}{T}\sum_{t=1}^Tu_tf_t'
$$
$$- H \hLam'\hSig_u^{-1}(S_u-\hSig_u)\hSig_u^{-1}\hLam H.$$
Note that  for $(\hLam,\hSig_u)\in\Xi_{\delta}$, $H=O_p(N^{-1})$, $J=O_p(1)$ for each element, $\|\hSig_u^{-1}\| =O_p(1)$, $\|\hLam\|_F=O_p(\sqrt{N})$,  hence 
$$\| \frac{1}{T}\sum_{t=1}^T f_t u_t' \hSig_u^{-1}\hLam H \|_F\leq O_p(1) \|\frac{1}{NT}\sum_{t=1}^Tf_tu_t'\|_F\|\hSig_u^{-1}\| \|\hLam\|_F=O_p(\frac{1}{N}\sqrt{\frac{N\log N}{T}}\sqrt{N})
=O_p(\sqrt{\frac{\log N}{T}})$$
Moreover, for the  empirical covariance $\|S_u\|^2\leq2\sum_{i,j\leq N}(T^{-1}\sum_{t=1}^Tu_{it}u_{jt}-\sigma_{u0,ij})^2+2\|\Sigun\|^2=O_p(T^{-1}N^2\log N+1)$ by Lemma \ref{la.1},  which implies $ H \hLam'\hSig_u^{-1}S_u\hSig_u^{-1}\hLam H =O_p(N^{-1}+T^{-1/2}(\log N)^{1/2})$. Also, $ H \hLam'\hSig_u^{-1}\hSig_u\hSig_u^{-1}\hLam H =H =O_p(N^{-1})$. Therefore $K=O_p(N^{-1}+T^{-1/2}(\log N)^{1/2}).$ It then implies (ii).

\end{proof}

\begin{lem} \label{la.4} Suppose Lemma \ref{assa.1} holds, then $J=o_p(1)$.
\end{lem}
\begin{proof}
By our assumption, both $\hLam'\hSig_u^{-1}\hLam$ and $\Lambda_0'\Sigun^{-1}\Lambda$ are diagonal. Moreover, the eigenvalues of $N^{-1}\hLam'\hSig_u^{-1}\hLam$ and $N^{-1}\Lambda_0'\Sigun^{-1}\Lambda$ are bounded away from zero. Therefore by Lemma \ref{assa.1}(i) and Lemma \ref{la.3}(ii),  there are two diagonal matrices $M_1$ and $M_2$ whose eigenvalues are all bounded away from zero, such that
\begin{eqnarray}
(I_r-J)M_1(I_r-J)'=M_2+o_p(1), \quad ( J -I_r)'( J -I_r)=I_r+o_p(1)
\end{eqnarray}
Applying  Lemma A.1 of Bai and Li (2012),  we have $J=o_p(1)$ and $M_1=M_2+o_p(1)$. We also assumed $\hLam$ and $\Lambda_0$ have the  same column signs, as a part of identification condition.

\end{proof}

\section{Proofs for Section 3}
In this section, $(\hLam,\hSig_u)=(\hLamo,\hSigo)$ and $
J=(\hLamo-\Lambda_0)(\hSigo)^{-1}\hLamo(\hLamop(\hSigo)^{-1}\hLamo)^{-1}.$  Throughout Appendix B, we will let $H=(\hLamop\hSigoi\hLamo)^{-1}$. For notational simplicity,  we let $$\omega_T=\frac{1}{\sqrt{N}}+\sqrt{\frac{\log N}{T}}.$$ 

We first cite a result from Fan et al. (2012):
\begin{thm}[Theorem 3.1 in Fan et al. (2012)] \label{thb.1}Suppose $(\log N)^{6/\gamma}=o(T)$ and $\sqrt{T}=o(N)$, then under Assumptions  3.1- 3.5, 
$$
\|\hSigo-\Sigun\|=O_p\left(m_N\omega_T^{1-q}\right)=\|(\hSigo)^{-1}-\Sigun^{-1}\|.
$$
\end{thm}
\begin{proof}
The sufficient conditions of this theorem are satisfied by our assumptions. See Fan et al. (2012).
\end{proof}

We then prove Lemma \ref{assa.1}, which then enables us to apply Lemmas \ref{la.3} and \ref{la.4}. Under Assumptions  3.1- 3.3, there is $\delta>0$ such that $(\Lambda_0,\Sigun)\in\Xi_{\delta}$ and $(\hLamo,\hSigo)\in\Xi_{\delta}$ with probability approaching one for $\Xi_{\delta}$ in Appendix A.  
\begin{lem}\label{lb.1}
For  $(\hLam,\hSig_u)=(\hLamo,\hSigo)$, Lemma  \ref{assa.1} is satisfied.
\end{lem}
\begin{proof}
The first order condition with respect to $\hLamo$ in  (ii)  is easy to verify, which is the same as that in Bai and Li (2012). We only show part (i).  

By definition, $L_1(\hLamo)\leq L_1(\Lambda_0)$. Also    the  representation  defined in Lemma \ref{la.2} yields
$$Q_3(\Lambda,\Sigma_u)+Q_2(\Lambda,\Sigma_u)=L_1(\Lambda)-N^{-1}\tr(S_u(\hSigo)^{-1})+N^{-1}\log|\hSigo|.$$
Thus
$$
Q_2( \hLamo, \hSigo)+Q_3(\hLamo, \hSigo)\leq Q_2( \Lambda_0, \hSigo)+Q_3(\Lambda_0, \hSigo)
$$
Note that $Q_2$ is always nonnegative and $Q_2( \Lambda_0, \hSigo)=0$. Therefore by Lemma \ref{la.2}, $0\leq Q_2( \hLamo, \hSigo)=o_p(1)$. Moreover, the matrix in the trace operation of $Q_2$ is semi-positive definite,  hence
\begin{equation}\label{eqb.1}
\frac{1}{N}\Lambda_0' (\hSigo)^{-1}\Lambda_0 -(I_r- J )\frac{1}{N}\hLamop(\hSigo)^{-1}\hLamo(I_r- J )'=o_p(1).
\end{equation} It remains to show that $N^{-1}\Lambda_0'((\hSigo)^{-1}-\Sigun^{-1})\Lambda_0=o_p(1)$, which follows immediately from Theorem \ref{thb.1} and that $m_N\omega_T^{1-q}=o(1)$.

\end{proof}

\subsection{Proof of Theorem 3.1}

\subsubsection{Consistency for $\hLamo$}
The equality  (\ref{eqb.1}) implies
$$
\frac{1}{N}(\hLamo-\Lambda_0)'(\hSigo)^{-1}(\hLamo-\Lambda_0)-\frac{1}{N}J\hLamop(\hSigo)^{-1}\hLamo J'=o_p(1).
$$ The second term is bounded by $N^{-1}\|J\|_F^2\|\hLamo\|_F^2\|\hSigoi\| =O_p(\|J\|_F^2)$. Lemma  \ref{la.4} then implies the second term is $o_p(1)$, which then implies that the first term is $o_p(1)$.   Because $\hSigoi$ has eigenvalues bounded away from zero asymptotically, we have $N^{-1}\|\hLamo-\Lambda_0\|_F^2=o_p(1)$.

\subsubsection{Consistency for $\hlam_j^{(1)}$}

Lemma \ref{la.3} (i) can be equivalently written as: for any $j\leq N$, 
\begin{equation}\label{eqb.2FOC}
\hlam_j^{(1)}-\lamj=-J'\lamj+H\hLamop(\hSigo)^{-1}a_j
\end{equation}
where $\hSig^{(1)}_{u,j}$ denotes the $j$th column of $\hSigo$, and $a_j$ is an $N\times 1$ vector $$
a_j= \Lambda_0T^{-1}\sum_{t=1}^Tf_tu_{jt}+T^{-1}\sum_{t=1}^T(u_tu_{jt}-\hSig^{(1)}_{u,j})+T^{-1}\sum_{t=1}^Tu_tf_t'\lamj
-\bar{u}\bar{u}_j.$$ 
The consistency of $\max_{j\leq N}\|\hlam_j^{(1)}-\lamj\|$ follows from Lemma \ref{la.4} and the following Lemma \ref{lb.2}.

\begin{lem}  \label{lb.2}$\max_{j\leq N}\|H\hLamop(\hSigo)^{-1}a_j\|=O_p(m_NN^{-1/2}\omega_T^{1-q}+T^{-1/2}(\log N)^{1/2}).$
\end{lem}
\begin{proof} By Lemma \ref{la.1},  uniformly in $j\leq N$, 
$$H\hLamop(\hSigo)^{-1}(\frac{1}{T}\sum_{t=1}^Tu_tf_t'\lamj+\Lambda_0\frac{1}{T}\sum_{t=1}^Tf_tu_{jt})=O_p( \frac{\sqrt{N}}{N}(2\sqrt{N}\sqrt{\frac{\log N}{T}}))=O_p(\sqrt{\frac{\log N}{T}}).$$
$$H\hLamop(\hSigo)^{-1}(\frac{1}{T}\sum_{t=1}^Tu_tu_{jt}-\Sigma_{u0,j})=O_p( \frac{\sqrt{N}}{N}\sqrt{N}\sqrt{\frac{\log N}{T}})=O_p(\sqrt{\frac{\log N}{T}}).$$
$$H\hLamop(\hSigo)^{-1}(\hSig_{u,j}^{(1)}-\Sigma_{u0,j})=O_p( \frac{\sqrt{N}}{N}m_N\omega_T^{1-q})=O_p( \frac{m_N}{\sqrt{N}}\omega_T^{1-q}).$$
Finally, $\max_{j\leq N}\|H\hLamo\hSigoi\bar{u}\bar{u}_j\|=O_p(\log N/T)$.
The result then follows from a triangular inequality and   that $m_N\omega_T^{1-q}=o(1)$.
\end{proof}

\subsection{Proof of Theorem  3.2}

\subsubsection{Uniform rate for $\hlam_j^{(1)}$}
By  (\ref{lb.2}), the uniform rate of convergence follows from Lemma \ref{lb.2} and the following Lemma \ref{lb.3}.

\begin{lem}\label{lb.3}
$J=O_p(m_N\omega_T^{1-q})$.
\end{lem}
\begin{proof}

The first order condition in Lemma \ref{la.3} (i) is equivalent to:
\begin{equation}\label{eqb.2}
J'J+J'+J+H\hLamop(\hSigo)^{-1}B(\hSigo)^{-1}\hLamo H=0
\end{equation}
where $B=\Lambda_0T^{-1}\sum_{t=1}^Tf_tu_t'+(\Lambda_0T^{-1}\sum_{t=1}^Tf_tu_t')'+S_u-\hSigo-\bar{u}\bar{u}'.$
We have, $\|\Lambda_0\|_F=O(\sqrt{N})$, $\bar{u}\bar{u}'=O_p(N\log N/T)$,  and $
\|S_u-\hSigo\|\leq \|\hSigo-\Sigma_{u0}\|+\|S_u-\Sigma_{u0}\|=O_p(NT^{-1/2}(\log N)^{1/2}+m_N\omega_T^{1-q}).
$ Therefore $H\hLamop(\hSigo)^{-1}B(\hSigo)^{-1}\hLamo H=O_p(T^{-1/2}(\log N)^{1/2}+m_NN^{-1}\omega_T^{1-q}).$
Since $J=o_p(1)$, $J'J$ can be ignored.  
It follows from (\ref{eqb.2}) that
\begin{equation}\label{eqb.3}
J'+J=O_p(\sqrt{\frac{\log N}{T}}+\frac{m_N\omega_T^{1-q}}{N}).
\end{equation}
Let $J_{ij}$ denote the $(i,j)$the entry of $J$. It then follows that $J_{ii}=O_p(T^{-1/2}(\log N)^{1/2}+m_NN^{-1}\omega_T^{1-q})$ for all $i\leq r.$ It is also not hard to verify that $\sqrt{(\log N)/T}=O(m_N\omega_T^{1-q})$ for any $0\leq q<1$ since $m_N\geq 1$.

On the other hand, due to the identification condition, both $\Lambda_0'\Sigun^{-1}\Lambda_0$ and $\hLamop(\hSigo)^{-1}\hLamo$ are diagonal. Let ndg$(M)$ denote the off-diagonal elements of $M$.  Then ndg$(\Lambda_0'\Sigun^{-1}\Lambda_0)=$ndg$(\hLamop(\hSigo)^{-1}\hLamo)=0$ is equivalent to
$$
\text{ndg}\{(\hLamo-\Lambda_0)'\hSigoi\hLamo+\hLamop\hSigoi(\hLamo-\Lambda_0)\}
$$
$$
=\text{ndg}\{-\Lambda_0'(\hSigoi-\Sigun^{-1})\Lambda_0+(\hLamo-\Lambda_0)'\hSigoi(\hLamo-\Lambda_0)\}
$$
Note that if $\text{ndg}\{M_1\}=\text{ndg}\{M_2\}$ then $\text{ndg}\{HM_1H\}=\text{ndg}\{HM_2H\}$ for two matrices $M_1$ and $M_2$ since $H$ is diagonal.  Also, $(\hLamo-\Lambda_0)'\hSigoi\hLamo H=J$.
The above identification condition implies
\begin{equation}\label{eqb.4}
\text{ndg}\{HJ+J'H\}=\text{ndg}\{-H\Lambda_0'((\hSigo)^{-1}-\Sigma_{u0}^{-1})\Lambda_0H+H(\hLamo-\Lambda_0)'\hSigoi(\hLamo-\Lambda_0)H\}
\end{equation}
 Note that $
H\Lambda_0'(\hSigoi-\Sigma_{u0}^{-1})\Lambda_0H=O_p( m_NN^{-1}\omega_T^{1-q}). $ 
Let $h_{ii}$ denote the $i$th diagonal entry of $H$. Let  $X=(\hLamo-\Lambda_0)'\hSigoi(\hLamo-\Lambda_0)$. Then  for $i\neq j$, (\ref{eqb.3}) and (\ref{eqb.4}) imply that
$$
J_{ji}+J_{ij}=O_p(\sqrt{\frac{\log N}{T}}+ \frac{m_N\omega_T^{1-q}}{N}),
$$
$$
h_{ii}J_{ij}+h_{jj}J_{ji}=O_p(\frac{m_N\omega_T^{1-q}}{N})+h_{ii}h_{jj}X_{ij}.
$$By assumption, with probability one, there is $\delta>0$ such that   $(N\delta)^{-1}<h_{ii}<N^{-1}\delta$, and $h_{ii}\neq h_{jj}$ for $i\neq j.$  Moreover, since all the eigenvalues of $\hSig_u$ are bounded away from zero and infinity, wpa1, $\|\hLam-\Lambda_0\|_F^2\geq c\|X\|_F$ for some $c>0.$ Then the above two equations imply that for any $i\neq j$, 
$J_{ij}=O_p(m_N\omega_T^{1-q})+O_p(N^{-1})X_{ji}$ (since $\sqrt{\log N/T}=O(m_N\omega_T^{1-q})$).
Then 
\begin{equation}\label{eqb.6addd}
\|J\|_F^2=O_p(m_N^2\omega_T^{2-2q}+\frac{1}{N^2}\|X\|_F^2).
\end{equation}
Moreover, by Lemma \ref{lb.2},  $\max_{j\leq N}\|H\hLamo(\hSigo)^{-1}a_j\|=O_p(m_N\omega_T^{1-q} )$.

We now show that $J=O_p(m_N\omega_T^{1-q})$.  Suppose this does not hold,  then (\ref{eqb.6addd}) implies $J=O_p(N^{-1}X)$.     By the definition 
$$
X=(\hLamo-\Lambda_0)'\hSigoi(\hLamo-\Lambda_0),
$$
$\|X\|_F=O_p(\|\hLamo-\Lambda_0\|_F^2)$. Therefore $J=O_p(N^{-1}X)$ yields  $\|J\|_F^2=O_p(N^{-2}\|\hLam-\Lambda_0\|_F^4)$. The first order condition (\ref{eqb.2FOC}) also yields  
$$
\max_{j\leq N}\|\hlam_j^{(1)}-\lambda_{0j}\|^2=O_p(\|J\|_F^2)=O_p(N^{-2}\|\hLamo-\Lambda_0\|_F^4),
$$
which implies $
\|\hLamo-\Lambda_0\|_F^2=\sum_{j=1}^N\|\hlam_j^{(1)}-\lambda_{0j}\|^2=O_p(N^{-1}\|\hLamo-\Lambda_0\|_F^4)$. Therefore
$$
\frac{1}{N^{-1}\|\hLamo-\Lambda_0\|_F^2}=\frac{\|\hLamo-\Lambda_0\|_F^2}{N^{-1}\|\hLamo-\Lambda_0\|_F^4}=O_p(1),
$$
which contradicts with the consistency $N^{-1}\|\hLamo-\Lambda_0\|_F^2=o_p(1)$.   This concludes the proof.

 \end{proof}
 Therefore, (\ref{eqb.2FOC})  gives  $\max_{j\leq N}\|\hlam_j^{(1)}-\lamj\|=O_p(\|J\|_F)=O_p(m_N\omega_T^{1-q})$. The rate of convergence for $N^{-1/2}\|\hLamo-\Lambda_0\|_F$ then follows immediately since it is bounded by $\max_{j\leq N}\|\hlam_j^{(1)}-\lamj\|$.

\subsection{Proof of Theorem  3.3}
By the definition of the covariance estimator in the first step, $\hSigo=(s_{ij}(R_{ij}))_{N\times N}$, where $s_{ij}$ is a chosen thresholding function. It was shown by Fan et al. (2012, Theorem 2.1) that $R_{ij}$  is the PCA estimator of $T^{-1}\sum_{t=1}^Tu_{it}u_{jt}$, that is, $R_{ij}=T^{-1}\sum_{t=1}^T\hu_{it}^{PCA}\hu_{jt}^{PCA}$.
 
\begin{lem}\label{lb.4} 
For any $\epsilon>0$, and any constant $M>0$, for all large enough $N, T$, $$
P(|R_{ij}|>M\tau_{ij}, \forall(i,j)\in S_U)>1-\epsilon.
$$
\end{lem}
\begin{proof}
We have,  $|R_{ij}|\geq |\Sigma_{u0, ij}|-|\Sigma_{u0, ij}-R_{ij}|$. Thus for all large enough $N,T$,
\begin{eqnarray*}
P(|R_{ij}|>M\tau_{ij}, \forall(i,j)\in S_U)&\geq& P(|\Sigma_{u0, ij}|>M\tau_{ij}+|\Sigma_{u0, ij}-R_{ij}|, \forall(i,j)\in S_U)\cr
&\geq& P(|\Sigma_{u0, ij}|/2>|\Sigma_{u0, ij}-R_{ij}|, \forall(i,j)\in S_U)>1-\epsilon,
\end{eqnarray*}
where in the second  and last inequalities we used the assumption that $\omega_T=o(\min_{(i,j)\in S_U}|\Sigma_{u0,ij}|)$ and the fact that $\max_{ij}|\Sigma_{u0, ij}-R_{ij}|=O_p(\omega_T)$.

\end{proof}

\textbf{Proof of Theorem  3.3}

By  Fan et al. (2012),  $\max_{i,j}|R_{ij}-\Sigma_{u0,ij}|=O_p(\omega_T)$, which implies for any $\epsilon>0$, there is  $C>0$ such that $P(\max_{i,j}|R_{ij}-\Sigma_{u0,ij}|>C\omega_T)<\epsilon/2$. For some universal $M>0$, we set the threshold $\tau_{ij}=M\alpha_{ij}\omega_T$ at entry $(i,j)$, where $\alpha_{ij}$ is  a data-dependent value that satisfies, for any $\epsilon>0$, there is $C_1>0$ such that $P(\alpha_{ij}>C_1, \forall i\neq j)>1-\epsilon/2.$ Then as long as  the constant $M$ in the definition of the threshold is larger than $2C/C_1$,   
$$P(\max_{i,j}|R_{ij}- \Sigma_{u0,ij}|>\min_{ij}\tau_{ij}/2)<P(\max_{i,j}|R_{ij}- \Sigma_{u0,ij}|>MC_1\omega_T/2)+\epsilon/2<\epsilon.$$
Note also that if $s_{ij}(R_{ij})\equiv\hSig_{ij}^{(1)}\neq0$,  then $|R_{ij}|>\tau_{ij}$, by the definition of $s_{ij}$. This implies,
$$
P(\hSig_{ij}^{(1)}\neq0,\exists (i,j)\in S_L)\leq P(|R_{ij}|>\tau_{ij},\exists (i,j)\in S_L)\leq P(\max_{(i,j)\in S_L}|R_{ij}|>\min_{ij}\tau_{ij})
$$
$$
\leq P(\max_{i,j}|R_{ij}-\Sigma_{u0,ij}|+\max_{(i,j)\in S_L}| \Sigma_{u0,ij}|>\min_{ij}\tau_{ij}).
$$
Since $\max_{(ij)\in S_L}|\Sigma_{u0,ij}|=o(\omega_T)$ by assumption,  for all large $T, N$
$$
P(\hSig_{ij}^{(1)}\neq0,\exists (i,j)\in S_L)\leq P(\max_{i,j}|R_{ij}- \Sigma_{u0,ij}|>\min_{ij}\tau_{ij}/2)<\epsilon.
$$
On the other hand,  for arbitrarily small $\epsilon>0$,   $P(\max_{ij}\tau_{ij}\leq K\omega_T)>1-\epsilon/2$ for some $K>0$,     which implies $P(|R_{ij}|\geq  M\omega_T+ K\omega_T, \forall (i,j)\in S_U)\leq  P(|R_{ij}|\geq  M\omega_T+ \tau_{ij}, \forall (i,j)\in S_U)+\epsilon/2.$  By the definition of $s_{ij}$, $|s_{ij}(z)-z|\leq\tau_{ij}$ for all $z\in\mathbb{R}$. Therefore  $|R_{ij}-\hSig_{u,ij}^{(1)}|=|R_{ij}-s_{ij}(R_{ij})|\leq \tau_{ij}$, hence for arbitrarily large $M>0$,  
$$
P(|\hSig_{u,ij}^{(1)}|>M\omega_T, \forall (i,j)\in S_U)\geq P(|R_{ij}|\geq  M\omega_T+ |R_{ij}-\hSig_{u,ij}^{(1)}|, \forall (i,j)\in S_U)
$$
$$
\geq P(|R_{ij}|\geq  (M+K)\omega_T, \forall (i,j)\in S_U)-\epsilon/2\geq 1-\epsilon
$$
where the last inequality follows from Lemma \ref{lb.4}.

\subsection{Proof of Theorems  3.4 and  3.5}

\subsubsection{Proof of Theorem  3.4}

A simple derivation implies that $\|N^{-1}\Lambda_0'(\hSigoi-\Siguni)\Lambda_0\|_F\leq N^{-1}\|\Lambda_0\|_F^2\|\hSigoi-\Siguni\|=O_p(m_N\omega_T^{1-q})$. This rate is not tight enough for the $\sqrt{T}$-consistency and limiting distribution  $\hlam_j^{(1)}$. A more refined rate  of  $N^{-1}\Lambda_0'(\hSigoi-\Siguni)\Lambda_0$ depends on the convergence properties of the PCA estimator.  We begin by citing some results proved by Fan et al. (2012). Recall that $R_{ij}$ denotes the $(i,j)$th entry of the orthogonal complement covariance in the sample covariance's spectrum decomposition, and $\hSig_{u,ij}^{(1)}=s_{ij}(R_{ij}).$

Let $\{\hu_{it}\}_{i\leq N, t\leq T}$  be the PCA estimates of $\{u_{it}\}_{i\leq N, t\leq T}$. Let $\lampca_j$ and   $\fpca$ denote the PCA estimators of the factor loadings and factors. 

\begin{lem}\label{lb.5}
(i) For any $i, j$, with probability one $R_{ij}=T^{-1}\sum_{t=1}^T\hu_{it}\hu_{jt}$,\\
(ii) $\max_{i\leq N}T^{-1}\sum_{t=1}^T(\hu_{it}-u_{it})^2=O_p(\omega_T^2).$\\
(iii) There is a nonsingular matrix $\bar{H}$ such that $T^{-1}\sum_{t=1}^T\|\fpca-\bar{H}f_t\|^2=O_p(T^{-1}+N^{-1})$ and $ \max_j\|\lampca_j-\bar{H}^{'-1}\lambda_{0j}\|=O_p(\omega_T).$\\
(iv) $\max_{i,j\leq N}|R_{ij}-\Sigma_{u0, ij}|=O_p(\omega_T).$
\end{lem}
\begin{proof}
See Theorem 2.1 and Lemma C.11 of Fan et al. (2012).
\end{proof}

\begin{lem}\label{lb.5add}
$
 \frac{1}{NT}\sum_{t=1}^T\sum_{i=1}^Nu_{it}\lambda_{0i}'\bar{H}^{-1}(\fpca-\bar{H}f_t) \xi_i\xi_i'=   O_p(\frac{1}{\sqrt{NT}}+\frac{1}{T}+\frac{1}{N})$.
\end{lem}
\begin{proof}
 By Bai (2003), there are two $r\times r$ matrices $\bar{H}$ and $V$, $\|V\|_F=O_p(1)$, $\|\bar{H}\|_F=O_p(1)$ such that 
 $\fpca-\bar{H}f_t=V(NT)^{-1}\sum_{s=1}^T\widehat{f}_s^{PCA}[u_s'u_t+f_s'\sum_{j=1}^N\lambda_{0j}u_{jt}+f_t'\sum_{j=1}^N\lambda_{0j}u_{js}]$. The desired result then follows from the following Lemma \ref{lb.51add}.

\end{proof}

\begin{lem}\label{lb.51add}
(i)  $\frac{1}{NT}\sum_{t=1}^T\sum_{i=1}^Nu_{it}\lambda_{0i}'\bar{H}^{-1}(NT)^{-1}\sum_{s=1}^T\widehat{f}_s^{PCA}u_s'u_t \xi_i\xi_i'= O_p(\frac{1}{\sqrt{NT}}+\frac{1}{T}+\frac{1}{N}) $\\
(ii)   $\frac{1}{NT}\sum_{t=1}^T\sum_{i=1}^Nu_{it}\lambda_{0i}'\bar{H}^{-1}(NT)^{-1}\sum_{s=1}^T\widehat{f}_s^{PCA}        f_s'\sum_{j=1}^N\lambda_{0j}u_{jt}      \xi_i\xi_i'= O_p(\frac{1}{\sqrt{NT}}+\frac{1}{N}) $\\
(iii) $\frac{1}{NT}\sum_{t=1}^T\sum_{i=1}^Nu_{it}\lambda_{0i}'\bar{H}^{-1}(NT)^{-1}\sum_{s=1}^T\widehat{f}_s^{PCA}         f_t'\sum_{j=1}^N\lambda_{0j}u_{js}      \xi_i\xi_i'=O_p(\frac{1}{\sqrt{NT}}+\frac{1}{T})$.
 
\end{lem}
\begin{proof}
 (i) We have,
\begin{eqnarray}
&&\|\frac{1}{NT}\sum_{t=1}^T\sum_{i=1}^Nu_{it}\frac{1}{NT}\sum_{s=1}^T\widehat{f}_s^{PCA'}u_s'u_t\bar{H}^{-1'}\lambda_{0i}\xi_i\xi_i'\|\leq \|\frac{1}{N^2T^2}\sum_{t=1}^T\sum_{i=1}^Nu_{it}\sum_{s=1}^Tf_s'\bar{H}'u_s'u_t\bar{H}^{-1'}\lambda_{0i}\xi_i\xi_i'\|\cr
&&+\|\frac{1}{N^2T^2}\sum_{t=1}^T\sum_{i=1}^Nu_{it}\sum_{s=1}^T(\widehat{f}_s^{PCA'}-f_s'\bar{H}')u_s'u_t\bar{H}^{-1'}\lambda_{0i}\xi_i\xi_i'\|=a+b.
\end{eqnarray}
We bound $a,b$ separately. Here $a$ is upper bounded by $a_1+a_2$, where by Cauchy-Schwarz,
\begin{eqnarray}\label{eqb.7add}
a_1&=&\|\frac{1}{N^2T^2}\sum_{t=1}^T\sum_{i=1}^Nu_{it}\sum_{s=1}^Tf_s'\bar{H}'(u_s'u_t-Eu_s'u_t)\bar{H}^{-1'}\lambda_{0i}\xi_i\xi_i'\|\cr
&\leq&\max_{i\leq N}\|\lambda_{0i}\xi_i\xi_i'\|(\frac{1}{T}\sum_{t=1}^Tu_{it}^2)^{1/2}\|\frac{1}{N}(\frac{1}{T}\sum_{t=1}^T\|\frac{1}{T}\sum_{s=1}^Tf_s(u_s'u_t-Eu_s'u_t)\|^2)^{1/2}\cr
&\leq&O_p(1)(\frac{1}{T}\sum_{t=1}^T\|\frac{1}{TN}\sum_{s=1}^Tf_s(u_s'u_t-Eu_s'u_t)\|^2)^{1/2}.
\end{eqnarray}
Note that $E\frac{1}{T}\sum_{t=1}^T\|\frac{1}{TN}\sum_{s=1}^Tf_s(u_s'u_t-Eu_s'u_t)\|^2=E\|\frac{1}{TN}\sum_{s=1}^Tf_s(u_s'u_t-Eu_s'u_t)\|^2$, which is $O(T^{-1}N^{-1})$ by Assumption  3.9. Hence $a_1=O_p((NT)^{-1/2})$.
\begin{equation}\label{eqb.8add}
a_2= \|\frac{1}{N^2T^2}\sum_{t=1}^T\sum_{i=1}^Nu_{it}\sum_{s=1}^Tf_s'\bar{H}'Eu_s'u_t\bar{H}^{-1'}\lambda_{0i}\xi_i\xi_i'\|\leq\max_{i\leq N}\frac{1}{T}\sum_{t=1}^T|u_{it}|O(1)\frac{1}{TN}\sum_{s=1}^T\|f_sEu_s'u_t\|
\end{equation}
Since $\max_{t\leq T}E(T^{-1}N^{-1}\sum_{s=1}^T\|f_sEu_s'u_t\|)\leq O(T^{-1})\max_t\sum_{s=1}^T|Eu_s'u_t|/N=O(T^{-1})$ by the strong mixing condition (Lemma C.5 of Fan Liao and Mincheva 2012), we have $a_2=O_p(T^{-1})$. This implies $a=O_p(N^{-1/2}T^{-1/2}+T^{-1})$.

Now we bound $b$. Using Cauchy Schwarz inequality, we have $b\leq b_1+b_2$ where
\begin{eqnarray}\label{eqb.9add}
b_1&=&\|\frac{1}{N^2T^2}\sum_{t=1}^T\sum_{i=1}^Nu_{it}\sum_{s=1}^T(\widehat{f}_s^{PCA'}-f_s'\bar{H}')(u_s'u_t-Eu_s'u_t)\bar{H}^{-1'}\lambda_{0i}\xi_i\xi_i'\|\cr
&&\leq O_p(1)\frac{1}{N^2T}\sum_{t=1}^T\sum_{i=1}^N|u_{it}|\left(\frac{1}{T}\sum_{s=1}^T\|\widehat{f}_s^{PCA}-\bar{H}f_s\|^2\right)^{1/2}\left(\frac{1}{T}\sum_{s=1}^T|u_s'u_t-Eu_s'u_t|^2\right)^{1/2}\cr
&\leq&O_p(\frac{1}{N})O_p(\frac{1}{\sqrt{T}}+\frac{1}{\sqrt{N}})O_p(\sqrt{N})=O_p(\frac{1}{N}+\frac{1}{\sqrt{NT}}),
\end{eqnarray}
where the second inequality follows from $ET^{-1}\sum_{s=1}^T|u_s'u_t-Eu_s'u_t|^2=O(N)$. Using Cauchy-Schwarz inequality, we also obtain
\begin{eqnarray}\label{eqb.10add}
b_2&=&\|\frac{1}{N^2T^2}\sum_{t=1}^T\sum_{i=1}^Nu_{it}\sum_{s=1}^T(\widehat{f}_s^{PCA'}-f_s'\bar{H}')(Eu_s'u_t)\bar{H}^{-1'}\lambda_{0i}\xi_i\xi_i'\|\cr
&\leq& O_p(1)(\frac{1}{T}\sum_{s=1}^T\|\widehat{f}_s^{PCA}-\bar{H}f_s\|^2)^{1/2}\left(\frac{1}{T}\sum_{s=1}^T|Eu_s'u_t/N|^2\right)^{1/2}\cr
&=&O_p(\frac{1}{\sqrt{NT}}+\frac{1}{T}).
\end{eqnarray}
(ii) Let $d_{i,kl}$  be the $(k,l)$th element of $\xi_i\xi_i'$. Then the $(k,l)$th element of the object of interest is bounded by $d_1+d_2$, where, by  Cauchy Schwarz inequality,
\begin{eqnarray}\label{eqb.11add}
d_1&=&|\frac{1}{(NT)^2}\sum_{t=1}^T\sum_{s=1}^T\sum_{j=1}^N\sum_{i=1}^N(u_{it}u_{jt}-Eu_{it}u_{jt})\lambda_{0i}'\bar{H}^{-1}\widehat{f}_s^{PCA}f_s'\lambda_{0j}d_{i,kl}|\cr
&\leq&O_p(1)(\frac{1}{T}\sum_{s=1}^T\|\widehat{f}_s^{PCA}\|^2)^{1/2}(\frac{1}{T}\sum_{s=1}^T\|\widehat{f}_s^{PCA}\|^2)^{1/2}\|\frac{1}{N^2T}\sum_{j=1}^N\sum_{i=1}^N\sum_{t=1}^T(u_{it}u_{jt}-Eu_{it}u_{jt})\lambda_{0i}\lambda_{0j}'d_{i,kl}\|\cr
&=&O_p(\frac{1}{\sqrt{NT}}).
\end{eqnarray} The last equality follows from Assumption  3.9.
Also, $\sum_{i,j\leq N}|Eu_{it}u_{jt}|=\sum_{(i,j)\in S_U}|\Sigma_{u0,ij}|+\sum_{(i,j)\in S_L}|\Sigma_{u0,ij}|=O(N)$. Thus
\begin{eqnarray}\label{eqb.12add}
d_2&=&|\frac{1}{(NT)^2}\sum_{t=1}^T\sum_{s=1}^T\sum_{j=1}^N\sum_{i=1}^N(Eu_{it}u_{jt})
\lambda_{0i}'\bar{H}^{-1}\widehat{f}_s^{PCA}f_s'\lambda_{0j}d_{i,kl}|\cr
&\leq&O_p(1)\frac{1}{N^2T}\sum_{s=1}^T\|\widehat{f}_s^{PCA}\|\|f_s\|\sum_{i,j\leq N}|Eu_{it}u_{jt}|=O_p(\frac{1}{N}).
\end{eqnarray}
(iii) The object of interest is bounded by $e_1+e_2$, where
\begin{equation}
e_1=\|\frac{1}{N^2T^2}\sum_{s=1}^T\sum_{j=1}^N\sum_{t=1}^T\sum_{i=1}^Nu_{it}\lambda_{0i}'\bar{H}^{-1}(\widehat{f}_s^{PCA}   -\bar{H}f_s)      f_t'\lambda_{0j}u_{js}      \xi_i\xi_i'
\|=O_p(\frac{1}{\sqrt{NT}}+\frac{1}{T}),
\end{equation}
and we used the fact that $\frac{1}{T}\sum_{t=1}^T\|\widehat{f}_t^{PCA}-\bar{H}f_t\|^2=O_p(T^{-1}+N^{-1})$ from Lemma \ref{lb.5}, and 
  that $N^{-1}\sum_{i=1}^N\|\frac{1}{T}\sum_{t=1}^Tf_tu_{it}\|=O_p(T^{-1/2}).$\footnote{We have $(N^{-1}\sum_{i=1}^N\|\frac{1}{T}\sum_{t=1}^Tf_tu_{it}\|)^2\leq N^{-1}\sum_{i=1}^N\|\frac{1}{T}\sum_{t=1}^Tf_tu_{it}\|^2=N^{-1}\sum_{i=1}^N\sum_{j=1}^r(\frac{1}{T}\sum_{t=1}^Tf_tu_{it})^2$, whose expectation is   $N^{-1}\sum_{i=1}^N\sum_{j=1}^r\var(\frac{1}{T}\sum_{t=1}^Tf_{jt}u_{it})$. Note that $\var(\frac{1}{T}\sum_{t=1}^Tf_{jt}u_{it})=O(T^{-1})$ uniformly in $i\leq N$. 
}
\begin{equation}
e_2=\|\frac{1}{N^2T^2}\sum_{s=1}^T\sum_{j=1}^N\sum_{t=1}^T\sum_{i=1}^Nu_{it}\lambda_{0i}'f_s      f_t'\lambda_{0j}u_{js}      \xi_i\xi_i'
\|=O_p(\frac{1}{T}).
\end{equation}

\end{proof}

\begin{lem}\label{lb.81add} For $S_U$ in the partition $\{(i,j): i,j\leq N\}=S_L\cup S_U$, \\
(i)  $\frac{1}{NT}\sum_{t=1}^T\sum_{i\neq j, (i,j)\in S_U}u_{it}\lambda_{0j}'\bar{H}^{-1}(NT)^{-1}\sum_{s=1}^T\widehat{f}_s^{PCA}u_s'u_t \xi_i\xi_j'= O_p(\frac{1}{\sqrt{NT}}+\frac{1}{T}+\frac{1}{N}) $\\
(ii)   $\frac{1}{NT}\sum_{t=1}^T\sum_{i\neq j, (i,j)\in S_U}u_{it}\lambda_{0j}'\bar{H}^{-1}(NT)^{-1}\sum_{s=1}^T\widehat{f}_s^{PCA}        f_s'\sum_{v=1}^N\lambda_{0v}u_{vt}      \xi_i\xi_j'= O_p(\frac{1}{\sqrt{NT}}+\frac{m_N}{N}) $\\
(iii) $\frac{1}{NT}\sum_{t=1}^T\sum_{i\neq j, (i,j)\in S_U}u_{it}\lambda_{0j}'\bar{H}^{-1}(NT)^{-1}\sum_{s=1}^T\widehat{f}_s^{PCA}         f_t'\sum_{v=1}^N\lambda_{0v}u_{vs}      \xi_i\xi_j'=O_p(\sqrt{\frac{\log N}{NT}}+\frac{\log N}{T})$.
 
\end{lem}
\begin{proof}
 (i) The term of interest is bounded by $a+b$, where
\begin{eqnarray*}
 a&=& \|\frac{1}{N^2T^2}\sum_{t=1}^T\sum_{(i,j)\in S_U, i\neq j}u_{it}\sum_{s=1}^Tf_s'\bar{H}'u_s'u_t\bar{H}^{-1'}\lambda_{0j}\xi_i\xi_j'\|,\cr
 b&=&\|\frac{1}{N^2T^2}\sum_{t=1}^T\sum_{(i,j)\in S_U, i\neq j}u_{it}\sum_{s=1}^T(\widehat{f}_s^{PCA'}-f_s'\bar{H}')u_s'u_t\bar{H}^{-1'}\lambda_{0j}\xi_i\xi_j'\|.
\end{eqnarray*}
Here $a$ is upper bounded by $a_1+a_2$, where\\  $a_1=\|\frac{1}{N^2T^2}\sum_{t=1}^T\sum_{(i,j)\in S_U, i\neq j}u_{it}\sum_{s=1}^Tf_s'\bar{H}'(u_s'u_t-Eu_s'u_t)\bar{H}^{-1'}\lambda_{0j}\xi_i\xi_j'\|$, and\\
 $a_2= \|\frac{1}{N^2T^2}\sum_{t=1}^T\sum_{(i,j)\in S_U, i\neq j}u_{it}\sum_{s=1}^Tf_s'\bar{H}'Eu_s'u_t\bar{H}^{-1'}\lambda_{0j}\xi_i\xi_j'\|.$ Note that $a_1$ and $a_2$ can be bounded in   the same way as (\ref{eqb.7add}) and (\ref{eqb.8add}). The only difference is that $N^{-1}\sum_{i=1}^N$ is replaced by a double sum $N^{-1}\sum_{(i,j)\in S_U, i\neq j}$. By the assumption, $N^{-1}\sum_{(i,j)\in S_U, i\neq j}1=O(1)$. The result of the proof is exactly the same, so is omitted. 
We conclude that  $a=O_p(N^{-1/2}T^{-1/2}+T^{-1})$.

On the other hand,     $b\leq b_1+b_2$ where\\
$
b_1=\|\frac{1}{N^2T^2}\sum_{t=1}^T \sum_{(i,j)\in S_U, i\neq j}u_{it}\sum_{s=1}^T(\widehat{f}_s^{PCA'}-f_s'\bar{H}')(u_s'u_t-Eu_s'u_t)\bar{H}^{-1'}\lambda_{0j}\xi_i\xi_j'\|$, and\\
$b_2=\|\frac{1}{N^2T^2}\sum_{t=1}^T\sum_{(i,j)\in S_U, i\neq j}u_{it}\sum_{s=1}^T(\widehat{f}_s^{PCA'}-f_s'\bar{H}')(Eu_s'u_t)\bar{H}^{-1'}\lambda_{0j}\xi_i\xi_j'\|$. Using Cauchy-Schwarz inequality and the strong mixing condition, $b_1$ and $b_2$ can be also bounded in an exactly the same way of (\ref{eqb.9add}) and (\ref{eqb.10add}). We conclude that $b=O_p(N^{-1}+T^{-1}+(NT)^{-1/2})$.

(ii) Let $d_{ij,kl}$  be the $(k,l)$th element of $\xi_i\xi_j'$. Then the $(k,l)$th element of the object of interest is bounded by $d_1+d_2$, where\\
$
d_1=|\frac{1}{(NT)^2}\sum_{t=1}^T\sum_{s=1}^T\sum_{(i,j)\in S_U, i\neq j}\sum_{v=1}^N(u_{it}u_{vt}-Eu_{it}u_{vt})\lambda_{0j}'\bar{H}^{-1}\widehat{f}_s^{PCA}f_s'\lambda_{0v}d_{ij,kl}|
$, and  \\
$d_2=|\frac{1}{(NT)^2}\sum_{t=1}^T\sum_{s=1}^T\sum_{(i,j)\in S_U, i\neq j}\sum_{v=1}^N(Eu_{it}u_{vt})\lambda_{0j}'\bar{H}^{-1}\widehat{f}_s^{PCA}f_s'\lambda_{0v}d_{ij,kl}|
$.  Bounding  $d_1, d_2$ is slightly different from (\ref{eqb.11add}) and (\ref{eqb.12add}), and we give the detail here. By Cauchy Schwarz inequality,
\begin{equation*}
d_1\leq O_p(1)(\frac{1}{T}\sum_{s=1}^T\|\widehat{f}^{PCA}_s\|^2)^{1/2}(\frac{1}{T}\sum_{s=1}^T\|{f}_s\|^2)^{1/2}\|\frac{1}{N^2T}\sum_{i\neq j, (i,j)\in S_U}\sum_{t=1}^T\sum_{v=1}^N(u_{it}u_{vt}-Eu_{it}u_{vt})\lambda_{0j}\lambda_{0v}'d_{ij, kl}\|
\end{equation*}
which is $O_p((NT)^{-1/2})$ by Assumption  3.9. On the other hand, \\$d_2\leq O_p(N^{-2})\sum_{i\neq j, (i,j)\in S_U}\sum_{k=1}^N|\Sigma_{u0,ik}|$. Note that $\|\Sigun\|_1=O(m_N)$, where $m_N$ is as defined in Assumption  3.1. Thus $d_2=O_p(N^{-1}m_N)$.

(iii) The object of interest is bounded by $e_1+e_2$, where\\
$
e_1=\|\frac{1}{N^2T^2}\sum_{s=1}^T\sum_{i\neq j, (i,j)\in S_U}\sum_{t=1}^T\sum_{v=1}^Nu_{it}\lambda_{0j}'\bar{H}^{-1}(\widehat{f}_s^{PCA}   -\bar{H}f_s)      f_t'\lambda_{0v}u_{vs}      \xi_i\xi_j'
\|,
$\\
$
e_1=\|\frac{1}{N^2T^2}\sum_{s=1}^T\sum_{i\neq j, (i,j)\in S_U}\sum_{t=1}^T\sum_{v=1}^Nu_{it}\lambda_{0j}'f_sf_t'\lambda_{0v}u_{vs}      \xi_i\xi_j'
\|
$. 

Since $\max_{i\leq N}\|T^{-1}\sum_{t=1}^Tf_tu_{it}\|=O_p(\sqrt{\log N/T})$, we conclude that $e_1=O_p(\frac{\sqrt{\log N}}{T}+\frac{\sqrt{\log N}}{\sqrt{NT}})$, and $e_2=O_p(\frac{\log N}{T})$.
\end{proof}

From Lemma  \ref{lb.81add}, immediately we have the following result.

\begin{lem}\label{lb.9add}
$
 \frac{1}{NT}\sum_{t=1}^T\sum_{i\neq j, (i,j)\in S_U}u_{it}\lambda_{0j}'\bar{H}^{-1}(\fpca-\bar{H}f_t) \xi_i\xi_j'=   O_p(\omega_T^2+m_N/N)$.
\end{lem}
\begin{proof}
Note that results (i)(ii)(iii) in Lemma  \ref{lb.81add} sum up to $O_p(\omega_T^2+m_N/N)$. Hence Lemma \ref{lb.9add} follows from the equality  $\fpca-\bar{H}f_t=V(NT)^{-1}\sum_{s=1}^T\widehat{f}_s^{PCA}[u_s'u_t+f_s'\sum_{j=1}^N\lambda_{0j}u_{jt}+f_t'\sum_{j=1}^N\lambda_{0j}u_{js}]$.  \end{proof}

The following lemma  strengthens the results of Bai (2003) when $\Sigun$ is sparse.

\begin{lem} \label{lb.6} For the PCA estimator, \\
(i) $N^{-1}\sum_{i=1}^N(R_{ii}-\Sigma_{u0,ii})\xi_i\xi_i'=O_p(\omega_T^2).$ \\
(ii) $N^{-1}\sum_{i\neq j, (i,j)\in S_U}(R_{ij}-\Sigma_{u0,ij})\xi_i\xi_j'=O_p(\omega_T^2+m_N/N)$.
\end{lem}

\begin{proof}
(i) $N^{-1}\sum_{i=1}^N(R_{ii}-\Sigma_{u0,ii})\xi_i\xi_i'=\sum_{i=1}^N(R_{ii}-S_{u,ii})\xi_i\xi_i'/N+\sum_{i=1}^N(S_{u,ii}-\Sigma_{u0,ii})\xi_i\xi_i'/N$.
By Assumption  3.9,  $\sum_{i=1}^N(S_{u,ii}-\Sigma_{u0,ii})\xi_i\xi_i'/N=\sum_{i=1}^N\sum_{t=1}^T
(u_{it}^2-\Sigma_{u0,ii})\xi_i\xi_i'/(NT)=O_p(1/\sqrt{NT}).$
On the other hand,  $\frac{1}{N}\sum_{i=1}^N(R_{ii}-S_{u,ii})\xi_i\xi_i'$ is equal to
\begin{eqnarray*}
\frac{1}{NT}\sum_{i=1}^N\sum_{t=1}^T(\hu_{it}^2-u_{it}^2)\xi_i\xi_i'=\frac{1}{NT}\sum_{i=1}^N\sum_{t=1}^T(\hu_{it}-u_{it})^2\xi_i\xi_i'+\frac{2}{NT}\sum_{i=1}^N\sum_{t=1}^Tu_{it}(\hu_{it}-u_{it})\xi_i\xi_i'.
\end{eqnarray*}
The first term on the right hand side is $O_p(\omega_T^2)$. We now work on  the second term.  By Bai (2003), there is a nonsingular   matrix $\bar{H}$ such that
\begin{equation}\label{eqPCA}
\hu_{jt}-u_{jt}=\lamj'\bar{H}^{-1}(\fpca-\bar{H}f_t)+(\lampca_j-\bar{H}^{'-1}\lamj)'(\fpca-\bar{H}f_t)+(\lampca_j-\bar{H}^{'-1}\lamj)\bar{H}f_t.
\end{equation}
By  Lemma \ref{lb.5add}
$
 \frac{1}{NT}\sum_{t=1}^T\sum_{i=1}^Nu_{it}\lambda_{0i}'\bar{H}^{-1}(\fpca-\bar{H}f_t)\xi_i\xi_i'=O_p(\frac{1}{\sqrt{NT}}+\frac{1}{T}+\frac{1}{N}).
  $
 In addition, for each element $d_{i,kl}$ of $\xi_i\xi_i'$, 
$$
\frac{1}{NT}\sum_{j=1}^N\sum_{t=1}^Tu_{jt}(\lampca_j-\bar{H}^{'-1}\lamj)\bar{H}f_td_{j,kl}\leq \frac{1}{N}\sum_{j=1}^N\|d_{j,kl}\frac{1}{T}\sum_{t=1}^Tu_{jt}f_t'\bar{H}'\| \max_j\|\lampca_j-\bar{H}^{'-1}\lambda_{0j}\|,$$
which is  $O_p(\omega_T\sqrt{\frac{\log N}{T}}).$ Also, 
$$
\frac{1}{NT}\sum_{j=1}^N\sum_{t=1}^Tu_{jt}(\lampca_j-\bar{H}^{'-1}\lamj)'(\fpca-\bar{H}f_t)d_{j,kl}=\frac{1}{T}\sum_{t=1}^T(\fpca-\bar{H}f_t)'\frac{1}{N}\sum_{j=1}^Nu_{jt}(\lampca_j-\bar{H}^{'-1}\lambda_{0j})d_{j,kl}$$
$$\leq\left(\frac{1}{T}\sum_{t=1}^T\|\fpca-\bar{H}f_t\|^2\max_j\|\lampca_j-\bar{H}^{'-1}\lambda_{0j}\|^2\frac{1}{T}\sum_{t=1}^T[\frac{1}{N}\sum_{j=1}^N|u_{jt}d_{j,kl}|]^2\right)^{1/2}=O_p(\frac{\omega_T}{\sqrt{T}}+\frac{\omega_T}{\sqrt{N}}).$$

(ii) Since $R_{ij}=T^{-1}\sum_{t=1}^T\hu_{it}\hu_{jt}$, the term of interest equals 
$$
\frac{2}{N}\sum_{i\neq j, (i,j)\in S_U}\frac{1}{T}\sum_{t=1}^Tu_{it}(\hu_{jt}-u_{jt})\xi_i\xi_j'+\frac{1}{N}\sum_{i\neq j, (i,j)\in S_U}\frac{1}{T}\sum_{t=1}^T(\hu_{it}-u_{it})(\hu_{jt}-u_{jt})\xi_i\xi_j'
$$
$$
+\frac{1}{NT}\sum_{i\neq j, (i,j)\in S_U}\sum_{t=1}^T(u_{it}u_{jt}-\Sigma_{u0,ij})\xi_i\xi_j'.
$$
By Assumption  3.9, the third term is $O_p((NT)^{-1/2})$.
By the assumption that $\sum_{i\neq j, (i,j)\in S_U}1=O(N)$ and Cauchy Schwarz inequality, the second term is $O_p(\omega_T^2)$. We now work out the first term.   Again we use the equality $\hu_{jt}-u_{jt}=\lamj'\bar{H}^{-1}(\fpca-\bar{H}f_t)+(\lampca_j-\bar{H}^{'-1}\lamj)'(\fpca-\bar{H}f_t)+(\lampca_j-\bar{H}^{'-1}\lamj)'\bar{H}f_t
$.   Lemma \ref{lb.9add} gives
$$
\frac{1}{N}\sum_{i\neq j, (i,j)\in S_U}\frac{1}{T}\sum_{t=1}^Tu_{jt}\lambda_{0i}'\bar{H}^{-1}(\fpca-\bar{H}f_t)\xi_i\xi_j'=O_p(\omega_T^2+\frac{m_N}{N}).
$$ 
On the other hand,
$
\frac{2}{N}\sum_{i\neq j, (i,j)\in S_U}\frac{1}{T}\sum_{t=1}^Tu_{jt}      (\lampca_i-\bar{H}^{-1'}\lambda_{0i})'\bar{H}f_t   \xi_i\xi_j'$ is bounded by,
$$ \max_{i\leq N}\|\xi_i\|^2\frac{1}{N}\sum_{i\neq j, (i,j)\in S_U}\|\frac{1}{T}\sum_{t=1}^Tu_{jt}f_t'\bar{H}'\| \max_j\|\lampca_j-\bar{H}^{'-1}\lambda_{0j}\|=O_p(\omega_T\sqrt{\frac{\log N}{T}}),$$
since $\max_{i\leq N}\|\xi_i\|=O(1)$. Also, $\frac{1}{N}\sum_{i\neq j, (i,j)\in S_U}\frac{1}{T}\sum_{t=1}^Tu_{jt}        (\lampca_i-\bar{H}^{-1'}\lambda_{0i})'  (\fpca-\bar{H}f_t)  \xi_i\xi_j'$ is bounded by
 $$
O(1)\max_{i\leq N}\|\hat{b}_i-H^{-1'}\lambda_{0i}\|\left(\frac{1}{T}\sum_{t=1}^T\|\hat{f}_t-\bar{H}f_t\|^2\right)^{1/2}\left(\frac{1}{T}\sum_{t=1}^T[\frac{1}{N}\sum_{i\neq j, (i,j)\in S_U}|d_{ij,kl}u_{it}|]^2\right)^{1/2}
$$
 which is $O_p(\frac{\omega_T}{\sqrt{T}}+\frac{\omega_T}{\sqrt{N}})$.

\end{proof}

\textbf{Proof of Theorem  3.4 $N^{-1}\Lambda_0'(\hSigoi-\Siguni)\Lambda_0=O_p(\omega_T^{2-2q}m_N^2)$}

\begin{proof}  By the triangular inequality, the left-hand-side is bounded by
$$
\frac{1}{N}\|\Lambda_0'(\hSigoi -\Sigma_{u0}^{-1})(\Sigma_{u0}-\hSigo)\Sigma_{u0}^{-1}\Lambda_0\|_F+\frac{1}{N}\|\Lambda_0'\Sigma_{u0}^{-1}(\Sigma_{u0}-\hSigo)\Sigma_{u0}^{-1}\Lambda_0\|_F.
$$
 The first term is $O_p(\omega_T^{2-2q}m_N^2  )$. We now bound the second term, which is
\begin{eqnarray*}
\frac{1}{N}\Xi(\hSigo-\Sigun)\Xi'&=&\frac{1}{N}\sum_{i=1}^N(R_{ii}-\Sigma_{u0,ii})\xi_i\xi_i'+\frac{1}{N}\sum_{i\neq j, (i,j)\in S_U}(\hSig_{u,ij}^{(1)}-\Sigma_{u0,ij})\xi_i\xi_j'\cr
&&+\frac{1}{N}\sum_{(i,j)\in S_L}(\hSig_{u,ij}^{(1)}-\Sigma_{u0,ij})\xi_i\xi_j',
 \end{eqnarray*}
 where $\Xi=\Lambda_0'\Sigun^{-1}.$
 The first  term  on the right hand side is $O_p(\omega_T^2)$ by Lemma \ref{lb.6}. The third term is dominated by, $
O( N^{-1})(\sum_{S_L}|\Sigma_{u0,ij}|+\sum_{S_L}|\hSig_{u,ij}^{(1)}|)=O(N^{-1})+O(N^{-1})\sum_{S_L}|\hSig_{u,ij}^{(1)}|.
 $
 By  Theorem  3.3, for any $\epsilon>0$ and any $M>0$, $
 P(\frac{1}{N}\sum_{(i,j)\in S_L}|\hSig_{u,ij}^{(1)}|>M\omega_T^2)\leq P(\exists(i,j)\in S_L, \hSig_{u,ij}^{(1)}\neq0)\leq\epsilon. $
 This implies the third term is $O_p(\omega_T^2).$  The second term equals
  $$
 \frac{1}{N}\sum_{i\neq j, (i,j)\in S_U}(\hSig_{u,ij}^{(1)}-R_{ij})\xi_i\xi_j'+ \frac{1}{N}\sum_{i\neq j, (i,j)\in S_U}(R_{ij}-\Sigma_{u0,ij})\xi_i\xi_j'.
 $$
 By Lemma  \ref{lb.6} (ii), $N^{-1}\sum_{i\neq j, (i,j)\in S_U}(R_{ij}-\Sigma_{u0,ij})\xi_i\xi_j'=O_p(\omega_T^2+m_N/N).$ On the other hand, recall that $|s_{ij}(z)-z|\leq a\tau_{ij}^2$ when $|z|>b\tau_{ij}$ (Section 3.1),
 \begin{eqnarray*}
&&\|\frac{1}{N}\sum_{i\neq j, (i,j)\in S_U}(\hSig_{u,ij}^{(1)}-R_{ij})\xi_i\xi_j'\|=\|\frac{1}{N}\sum_{i\neq j, (i,j)\in S_U, | R_{ij}|>b\tau_{ij}}(s_{ij}(R_{ij})-R_{ij})\xi_i\xi_j'\cr
&&+\frac{1}{N}\sum_{i\neq j, (i,j)\in S_U, | R_{ij}|>b\tau_{ij}}(s_{ij}(R_{ij})-R_{ij})\xi_i\xi_j'\|\leq O_p(\omega_T^2)+\|\frac{1}{N}\sum_{i\neq j, (i,j)\in S_U, | R_{ij}|\leq b\tau_{ij}}(s_{ij}(R_{ij})-R_{ij})\xi_i\xi_j'\|.
 \end{eqnarray*}
 Write $v=\|N^{-1}\sum_{i\neq j, (i,j)\in S_U, | R_{ij}|\leq b\tau_{ij}}(s_{ij}(R_{ij})-R_{ij})\xi_i\xi_j'\|$, then for any $C>0$, and $\epsilon>0$, Lemma \ref{lb.4} implies $
 P(v>M\omega_T^2)\leq P(\exists (i,j)\in S_U, |R_{ij}|\leq b\tau_{ij})<\epsilon$, which yields $v=O_p(\omega_T^2)$. Therefore $\frac{1}{N}\sum_{i\neq j, (i,j)\in S_U}(\hSig_{u,ij}^{(1)}-\Sigma_{u0,ij})\xi_i\xi_j'=O_p(\omega_T^2)$. This implies $N^{-1}\Lambda_0'(\hSigoi-\Siguni)\Lambda_0=O_p(\omega_T^2+\omega_T^{2-2q}m_N^2+m_N/N)=O_p(\omega_T^{2-2q}m_N^2)$.

\end{proof}

\subsubsection{Convergence rate for $J$}
We now improve the rate in Lemma \ref{lb.3}.

\begin{lem}\label{lb.8}
(i) $  H\hLamop\hSigoi[\Lambda_0\frac{1}{T}\sum_{t=1}^Tf_tu_t'+(\Lambda_0\frac{1}{T}\sum_{t=1}^Tf_tu_t')']\hSigoi\hLamo H=O_p(m_NT^{-1/2}(\log N)^{1/2}\omega_T^{1-q}).$\\
(ii) $H\hLamop\hSig_u^{-1}(S_u-\hSigo)\hSigoi\hLamo H=O_p(m_N\omega_T^{1-q}T^{-1/2}(\log N)^{1/2}+m_N\omega_T^{1-q}N^{-1}).$
\end{lem}
\begin{proof}
(i) By Theorem  3.2,  
\begin{equation}\label{eqb.6add}
\|\hLamo-\Lambda_0\|_F=O_p(\sqrt{N}m_N\omega_T^{1-q})=\|\hLamop\hSigoi-\Lambda_0'\Sigun^{-1}\|_F.
\end{equation}
 Therefore the RHS of part (i) equals
\begin{equation}
H\Lambda_0'\Siguni\left[\Lambda_0\frac{1}{T}\sum_{t=1}^Tf_tu_t'+(\Lambda_0\frac{1}{T}\sum_{t=1}^Tf_tu_t')'\right]\Siguni\Lambda_0 H+O_p(m_N\sqrt{\frac{\log N}{T}}\omega_T^{1-q}).
\end{equation}
Now it follows from Assumption  3.10 that
$$
\frac{1}{NT}\sum_{t=1}^Tf_tu_t' \Sigma_{u0}^{-1}\Lambda_0=\frac{1}{NT}\sum_{t=1}^T\sum_{i=1}^Nf_tu_{it}\xi_i'=O_p(\frac{1}{\sqrt{NT}})=O_p(m_N\sqrt{\frac{\log N}{T}}\omega_T^{1-q}),
$$
which then yields the desired result.

(ii)  Recall that $\|S_u-\hSigo\|=O_p(NT^{-1/2}(\log N)^{1/2}+m_N\omega_T^{1-q})$ and that 

$\|\hLamop\hSigoi-\Lambda_0'\Sigun^{-1}\|_F=O_p(\sqrt{N}m_N\omega_T^{1-q})$. 
By Theorem \ref{thb.1}, the RHS of (ii) equals
$$
H\Lambda_0'\Siguni(S_u-\Sigun)\Siguni\Lambda_0 H+O_p(m_N\omega_T^{1-q}\sqrt{\frac{\log N}{T}}+\frac{m_N\omega_T^{1-q}}{N}).
$$
By Assumption  3.10 (note that $H=O_p(N^{-1})$), 
$$H\Lambda_0'\Siguni(S_u-\Sigun)\Siguni\Lambda_0 H=\frac{1}{T}H\sum_{i\leq N,j\leq N}\sum_{t\leq T}(u_{it}u_{jt}-Eu_{it}u_{jt})\xi_i\xi_j'H=O_p(\frac{1}{\sqrt{NT}}).$$

\end{proof}

\begin{lem}\label{lb.9}
$J=O_p(m_N^2\omega_T^{2-2q})$.
\end{lem}

\begin{proof}
By (\ref{eqb.2}) and Lemma \ref{lb.8}, ignoring the smaller order $J'J$, we have
$$
J+J'=O_p(m_N\omega_T^{1-q}\sqrt{\frac{\log N}{T}}+\frac{m_N\omega_T^{1-q}}{N}).
$$  This implies that $J_{ii}=O_p(m_N\omega_T^{1-q}(T^{-1/2}(\log N)^{1/2}+N^{-1})).$ 

Moreover, since $H(\hLamo-\Lambda_0)'\hSigoi(\hLamo-\Lambda_0)H=O_p(N^{-1}m_N^2\omega_T^{2-2q})$, (\ref{eqb.4}) and Theorem  3.4 imply $\text{ndg}\{HJ+J'H\}=O_p(N^{-1}m_N^2\omega_T^{2-2q})$. Therefore for $i\neq j$, $J_{ij}=O_p(m_N^2\omega_T^{2-2q}+m_N\omega_T^{1-q}(T^{-1/2}(\log N)^{1/2}+N^{-1}))=O_p(m_N^2\omega_T^{2-2q}).$
The desired result follows immediately.
\end{proof}

\subsubsection{Improved rate for $\hlam_j^{(1)}$}

\begin{lem}\label{lb.10}
(i) $H\hLamop\hSigoi T^{-1}\sum_{t=1}^T(u_{t}u_{jt}-\hSig_{u,j}^{(1)})=O_p(m_N\omega_T^{2-q}+m_N^2\omega_T^{2-2q}N^{-1/2})$.\\
(ii) $H\hLamop\hSigoi T^{-1}\sum_{t=1}^Tu_tf_t'\lamj=O_p(m_N\omega_T^{1-q}T^{-1/2}(\log N)^{1/2}).$
\end{lem}
\begin{proof}
(i) We have, $\|\hLamop\hSigoi-\Lambda_0'\Sigun^{-1}\|_F=O_p(\sqrt{N}m_N\omega_T^{1-q})$.  Hence $
H(\hLamop\hSigoi-\Lambda_0'\Sigun^{-1})T^{-1}\sum_{t=1}^T(u_{t}u_{jt}-\hSig_{u,j}^{(1)})=O_p(m_N\omega_T^{1-q}T^{-1/2}(\log N)^{1/2}+N^{-1/2}m_N^2\omega_T^{2-2-q}).
$ Hence part (i) equals
$$
H\Lambda_0'\Sigun^{-1}T^{-1}\sum_{t=1}^T(u_{t}u_{jt}-E(u_{t}u_{jt}))+O_p(m_N\omega_T^{2-q}+\frac{m_N^2\omega_T^{2-2q}}{\sqrt{N}})
$$
where $O_p(\cdot)$ is uniform in $j\leq N.$ By Assumption  3.10,  for each $j\leq N$, $$H\Lambda_0'\Sigun^{-1}T^{-1}\sum_{t=1}^T(u_{t}u_{jt}-E(u_{t}u_{jt}))=H\frac{1}{T}\sum_{i=1}^N\sum_{t=1}^T\xi_i(u_{it}u_{jt}-E(u_{t}u_{jt}))=O_p((NT)^{-1/2}).$$

(ii) We have $\|H(\hLamop\hSigoi-\Lambda_0'\Sigun^{-1})T^{-1}\sum_{t=1}^Tu_{t}f_t'\|_F=O_p(m_N\omega_T^{1-q}T^{-1/2}(\log N)^{1/2})$. Hence (ii) equals
$$
H\Lambda_0'\Sigun^{-1}\frac{1}{T}\sum_{t=1}^Tu_tf_t'\lamj+O_p(m_N\omega_T^{1-q}\sqrt{\frac{\log N}{T}}).
$$
By Assumption  3.10,  the first term equals $NH(NT)^{-1}\sum_{i=1}^N\sum_{t=1}^T\xi_iu_{it}f_t'\lamj=O_p((NT)^{-1/2})$, which yields the desired result.
\end{proof}

\begin{lem}\label{lb.11}
 For each fixed $j\leq N$, $$\hlam_j^{(1)}-\lamj=H\hLamop\hSigoi\Lambda_0\frac{1}{T}\sum_{t=1}^Tf_tu_{jt}+ O_p(m_N^2\omega_T^{2-2q}).$$
\end{lem}
\begin{proof}
Note that those two terms in Lemma \ref{lb.10} (i) (ii) are dominated by $O_p(m_N^2\omega_T^{2-2q})$. Therefore, the desired expansion follows from  the first order condition (\ref{eqb.2FOC}) and Lemma \ref{lb.9}.
\end{proof}

\subsubsection{Proof of Theorem  3.5}
By Lemma \ref{lb.11}, and (\ref{eqb.6add})
 \begin{eqnarray*}\hlam_j^{(1)}-\lamj&=&H\hLamop\hSigoi(\Lambda_0-\hLamo)\frac{1}{T}\sum_{t=1}^Tf_tu_{jt}+ \frac{1}{T}\sum_{t=1}^Tf_tu_{jt}+O_p(m_N^2\omega_T^{2-2q})\cr
 &=& \frac{1}{T}\sum_{t=1}^Tf_tu_{jt}+O_p(m_N^2\omega_T^{2-2q}+m_N\omega_T^{1-q}\sqrt{\frac{\log N}{T}})=\frac{1}{T}\sum_{t=1}^Tf_tu_{jt}+O_p(m_N^2\omega_T^{2-2q}).
 \end{eqnarray*}
By the assumption that $m_N^2\omega_T^{2-2q}=o(T^{-1/2})$, we have $\sqrt{T}(\hlam_j^{(1)}-\lamj)=T^{-1/2}\sum_{t=1}^Tf_tu_{jt}+o_p(1)$. The limiting distribution follows since
$$
T^{-1/2}\sum_{t=1}^Tf_tu_{jt}\rightarrow^d N_r(0,E(u_{jt}f_tf_t')).
$$

\subsection{Proof of Theorem  3.6}
 
 For any $t\leq T$, $y_t-\bar{y}=\Lambda_0f_t+u_t-\bar{u}$. Hence 
\begin{equation}\label{eqb.9}
\hfto-f_t=-J'f_t+(\hLamop\hSigoi\hLamo)^{-1}\hLamop\hSigoi(u_t-\bar{u}).
\end{equation}

\subsubsection{Convergence rate}
 Since both $f_t$ and $u_t$ have exponential tails, using Bonferroni's method we have, $\max_t\|f_t\|=O_p((\log T)^{1/r_2})$ and $\max_t\|u_t\|=O_p(\sqrt{N}(\log T)^{1/r_1})$.  Thus by Lemma \ref{lb.9}, $\max_{t\leq T}\|J'f_t\|=O_p(m_N^2\omega_T^{2-2q}(\log T)^{1/r_2})$. The term with $\bar{u}$ in (\ref{eqb.9}) is of smaller order hence is negligible. Also $
\|(\hLamop\hSigoi\hLamo)^{-1}(\hLamop\hSigoi-\Lambda_0'\Sigun^{-1})\|_F=O_p(N^{-1/2}m_N\omega_T^{1-q})$,  where we used  $\|\hLamop\hSigoi-\Lambda_0'\Sigun^{-1}\|_F=O_p(\sqrt{N}m_N\omega_T^{1-q})$.  Hence
$$
\max_{t\leq T}(\hLamop\hSigoi\hLamo)^{-1}\hLamop\hSigoi(u_t-\bar{u})=O_p(\frac{1}{N})\Lambda_0'\Sigun^{-1}u_t$$
$$+O_p(m_N^2\omega_T^{2-2q}(\log T)^{1/r_2}+m_N\omega_T^{1-q}(\log T)^{1/r_1})=O_p(\frac{1}{N})\Lambda_0'\Sigun^{-1}u_t+O_p(m_N\omega_T^{1-q}(\log T)^{1/r_1+1/r_2}).
$$
Finally, because $E(\frac{1}{N}\Lambda_0'\Sigun^{-1}u_tu_t'\Sigun^{-1}\Lambda_0)=\frac{1}{N}\Lambda_0'\Sigun^{-1}\Lambda_0$,  whose eigenvalues are bounded. Hence $\frac{1}{\sqrt{N}}\Lambda_0'\Sigun^{-1}u_t=O_p(1)$. Also,  $O_p(N^{-1/2})$ is of smaller order than\\ $O_p(m_N\omega_T^{1-q}(\log T)^{1/r_1+1/r_2+1})$. This implies
$$
\|\hfto-f_t\|=O_p(m_N\omega_T^{1-q}(\log T)^{1/r_1+1/r_2+1}).
$$
The above proof also shows that the rate can be made uniform if $\max_{t\leq T}\|\frac{1}{\sqrt{N}}\Lambda_0'\Sigun^{-1}u_t\|=O_p(\log T).$

\subsubsection{Asymptotic normality}

Recall that $\Xi=\Lambda_0'\Sigun^{-1}$ and $\beta_t=\Sigun^{-1}u_t$.
\begin{lem}\label{lb.12} For any fixed $t\leq T$,
$N^{-1/2}(\hLamo-\Lambda_0)'\Sigun^{-1} u_t=o_p(1)$.
\end{lem}
\begin{proof}  
We expand $\hLamo-\Lambda_0$ using the first order condition
\begin{equation}
(\hLamo-\Lambda_0)'=J\Lambda_0'+H\hLamop\hSigoi[\Lambda_0\frac{1}{T}\sum_{s=1}^Tf_su_s'+\frac{1}{T}\sum_{s=1}^Tu_sf_s'\Lambda_0'+S_u -\hSigo]
\end{equation}
 and investigate each term separately.  First of all, since $J=O_p(m_N^2\omega_T^{2-2q})$, and by assumption that $\Lambda_0'\Sigun^{-1}u_t=\sum_{i=1}^N\xi_iu_{it}=O_p(\sqrt{N})$, we have
$N^{-1/2}J\Lambda_0'\Sigun^{-1}u_t=O_p(m_N^2\omega_T^{2-2q}).$ Second,  by the assumption that $(TN)^{-1/2}\sum_{s=1}^Tf_su_s'\Sigun^{-1}u_t=O_p(1)$, we have
$$
\frac{1}{\sqrt{N}}H\hLamop\hSigoi\Lambda_0\frac{1}{T}\sum_{s=1}^Tf_su_s'\Sigun^{-1}u_t=O_p(\frac{1}{\sqrt{T}}).
$$
Third, $N^{-1/2}H\hLamop\hSigoi\frac{1}{T}\sum_{s=1}^Tu_sf_s'\Lambda_0'\Sigun^{-1}u_t=O_p(\sqrt{\log N/T}).$ Moreover, \\ $N^{-1/2}H(\hLamop\hSigoi-\Lambda_0'\Sigun^{-1})(S_u-\Sigun)\Sigun^{-1}u_t=O_p(m_N\omega_T^{1-q}\sqrt{N\log N/T})=o_p(1)$. Therefore,  by the assumption that $(NT\sqrt{N})^{-1}\sum_{i=1}^N\sum_{s=1}^T\xi_i(u_{is}u_s'-Eu_{is}u_s')\beta_t=o_p(1)$, we have,
\begin{eqnarray*}
&&\frac{1}{\sqrt{N}}H\hLamop\hSigoi(S_u-\Sigun)\Sigun^{-1}u_t=\frac{1}{\sqrt{N}}H\Lambda_0'\Sigun^{-1}(S_u-\Sigun)\Sigun^{-1}u_t+o_p(1)\cr
&&=\frac{1}{T\sqrt{N}}H\sum_{i=1}^N\sum_{j=1}^N\sum_{s=1}^T\xi_i(u_{is}u_{js}-Eu_{is}u_{js})\beta_{jt}=o_p(1).
\end{eqnarray*}
Finally, $N^{-1/2}H\hLamop\hSigoi(\hSigo-\Sigun)\Sigun^{-1}u_t=O_p(\frac{1}{\sqrt{N}}m_N\omega_N^{1-q}).$
\end{proof}

\begin{lem}\label{lb.13} For any fixed $t\leq T$,
$N^{-1/2}\Lambda_0'(\hSigoi-\Sigun^{-1})u_t=o_p(1)$
\end{lem}
\begin{proof}  We note that, 
$N^{-1/2}\Lambda_0'(\hSigoi-\Sigun^{-1})u_t=N^{-1/2}\Xi(\hSigo-\Sigun)\beta_t+
O_p(\sqrt{N}m_N^2\omega_T^{2-2q})$. On the other hand,
\begin{eqnarray*}
\frac{1}{\sqrt{N}}\Xi(\hSigo-\Sigun)\beta_t&=&\frac{1}{\sqrt{N}}\sum_{i=1}^N(R_{ii}-\Sigma_{u0,ii})\xi_i\beta_{it}+\frac{1}{\sqrt{N}}\sum_{i\neq j, (i,j)\in S_U}(\widehat{\Sigma}_{u,ij}^{(1)}-\Sigma_{u0,ij})\xi_i\beta_{jt}\cr
&&+\frac{1}{\sqrt{N}}\sum_{  (i,j)\in S_L}(\widehat{\Sigma}_{u,ij}^{(1)}-\Sigma_{u0,ij})\xi_i\beta_{jt}.
\end{eqnarray*}
The result of the proof is very similar to  that of Lemmas \ref{lb.6} and Theorem {l3.1}, based on the expansion (\ref{lb.6}) and Theorem  3.3, hence is omitted.
\end{proof}

\textbf{Proof of asymptotic normality}

We now fix $t$, then Lemma \ref{lb.9} gives $J'f_t=O_p(m_N^2\omega_T^{2-2q}).$ Hence $\sqrt{N}J'f_t$ is negligible as $\sqrt{N}m_N^2\omega_T^{2-2q}=o(1).$ Moreover, $(\hLamop\hSigoi\hLamo)^{-1}\hLamop\hSigoi\bar{u}$ is of smaller order of $(\hLamop\hSigoi\hLamo)^{-1}\hLamop\hSigoi u_t$, hence is negligible. Next,  $$\sqrt{N}(\hLamop\hSigoi\hLamo)^{-1}\hLamop\hSigoi u_t=\sqrt{N}(\Lambda_0'\Sigun^{-1}\Lambda_0)^{-1}\Lambda_0'\Sigun^{-1}u_t$$ 
$$
+O_p(N^{-1/2})(\hLamop\hSigoi-\Lambda_0'\Sigun^{-1})u_t+O_p(m_N\omega_T^{1-q})
$$
where we used $(\hLamop\hSigoi\hLamop)^{-1}-(\Lambda_0'\Sigun^{-1}\Lambda_0)^{-1}=O_p(N^{-1}m_N\omega_T^{1-q})$. By Lemmas \ref{lb.12} and \ref{lb.13}, $N^{-1/2}(\hLamop\hSigoi-\Lambda_0'\Sigun^{-1})u_t=o_p(1)$. This implies,  for each fixed $t$,
\begin{eqnarray*}
\sqrt{N}(\hfto-f_t)&=&\sqrt{N}(\Lambda_0'\Sigun^{-1}\Lambda_0)^{-1}\Lambda_0'\Sigun^{-1}u_t+O_p(\sqrt{N}m_N^2\omega_T^{2-2q}
+m_N\omega_T^{1-q})\cr
&=&\sqrt{N}(\Lambda_0'\Sigun^{-1}\Lambda_0)^{-1}\Lambda_0'\Sigun^{-1}u_t+o_p(1).
\end{eqnarray*}
The asymptotic normality then follows from the fact that 
$$N^{-1/2}\Lambda_0'\Sigun^{-1}u_t=\frac{1}{\sqrt{N}}\sum_{i=1}^N\xi_iu_{it}\rightarrow^dN(0,Q).$$


\section{Proofs of Section 4}

\subsection{Proof of Theorem  4.1}

Define
\begin{eqnarray*}
Q_1(\Sigma_u)&=&\frac{1}{N}\log|\Sigma_u|+\frac{1}{N}\tr(S_u\Sigma_u^{-1})+\frac{\mu_T}{N}\sum_{i\neq j} w_{ij}|\Sigma_{u, ij}|\cr
&&-\frac{1}{N}\log|\Sigma_{u0}|-\frac{1}{N}\tr(S_u\Sigma_{u0}^{-1})-\frac{\mu_T}{N}\sum_{i\neq j} w_{ij}|\Sigma_{u0, ij}|,
\end{eqnarray*}
Let $L_c(\Lambda,\Sigma_u)=L_2(\Lambda,\Sigma_u)-N^{-1}\log|\Sigma_{u0}|-N^{-1}\tr(S_u\Sigma_{u0}^{-1})-N^{-1}\mu_T\sum_{i\neq j} w_{ij}|\Sigma_{u0, ij}|$. Then the minimizer of $L_c$ is the same as that of $L_2.$ This implies $L_c(\hLamt,\hSigt)\leq L_c(\Lambda_0,\Sigun)$. Recall the definitions of $Q_2(\Lambda,\Sigma_u)$ and $Q_3(\Lambda,\Sigma_u).$ Then
$$
L_c(\Lambda, \Sigma_u)=Q_1(\Sigma_u)+Q_2(\Lambda,\Sigma_u)+Q_3(\Lambda,\Sigma_u).
$$

\begin{lem} \label{lc.1} There is a nonnegative stochastic sequence $0\leq d_T=O_p(
N^{-1}\log N+T^{-1/2}(\log N)^{1/2})$ such that  $Q_1(\hSigt)\leq d_T$ with probability one.
\end{lem}

\begin{proof} We have $Q_2(\hLamt, \hSigt)\geq0$. In addition, $Q_2(\Lambda_0, \Sigma_{u0})=Q_1(\Sigma_{u0})=0$.  Hence
\begin{eqnarray*}
Q_1(\hSigt)&=&L_c(\hLamt,\hSigt)-Q_2(\hLamt,\hSigt)-Q_3(\hLamt,\hSigt)\cr
&\leq& L_c(\hLamt,\hSigt)-Q_3(\hLamt,\hSigt)\leq L_c(\Lambda_0,\Sigma_{u0})-Q_3(\hLamt,\hSigt)\cr
&=&Q_3(\Lambda_0,\Sigma_{u0})-Q_3(\hLamt,\hSigt).
\end{eqnarray*}
By the definition of $\Theta_{\lambda}\times \Gamma$, there is $\delta>0$ such that $\Theta_{\lambda}\times \Gamma\subset\Xi_{\delta}.$ The result then holds for  $d_T=|Q_3(\Lambda_0,\Sigma_{u0})|+|Q_3(\hLamt,\hSigt)|$ by Lemma \ref{la.2}. 

\end{proof}

Throughout, let  (recall that $D=\sum_{i\neq j, (i,j)\in S_U}1.$)
$$\Delta= \hSigti-\Sigma_{u0}^{-1}, \hspace{1em} K_T=\sum_{(i,j)\in S_L}|\Sigma_{u0, ij}|.$$

\begin{lem} \label{lc.2} For all large enough $T$ and $N$,
\begin{eqnarray*}
NQ_1(\hSigt)
&\geq&\frac{1}{2}\mu_T\min_{(i,j)\in S_L}w_{ij}\sum_{(i,j)\in S_L}|\hSig_{u, ij}-\Sigma_{u0, ij}|+c\|\Delta\|_F^2-2\mu_T\max_{(i,j)\in S_L}w_{ij}K_T\cr
&&-\left(O_p(\sqrt{\frac{\log N}{T}})\sqrt{N+D}+\mu_T\max_{i\neq j, (i,j)\in S_U}w_{ij}\sqrt{D}\right)\|\Delta\|_F.
\end{eqnarray*}

\end{lem}

\begin{proof} Let $\Omega_0=\Sigma_{u0}^{-1}$, $\hO=\hSigti$.  For any $\Sigma_u$, let $\Omega =\Sigma_u^{-1}$. Define a function  $f(t)=-\log|\Omega _0+t\Delta|+\tr(S_u(\Omega _0+t\Delta))$, $t\geq0.$ Then $
-\log|\hO|+\tr(S_u\hO)=f(1);  $  $ -\log|\Omega _0|+\tr(S_u\Omega _0)=f(0);
$
and 
\begin{equation}\label{ea.9}
NQ_1(\hSigt)=f(1)-f(0)+{\mu_T}\sum_{i\neq j} w_{ij}|\hSig_{u, ij}|-\mu_T\sum_{i\neq j} w_{ij}|\Sigma_{u0, ij}|
\end{equation}
By the integral remainder Taylor expansion,  $f(1)-f(0)=f'(0)+\int_0^1(1-t)f''(t)dt$. We now calculate $f'(0)$ and $f''(t)$. Using the matrix differentiation formula, we have, $
f'(t)=\tr(S_u\Delta)-\tr((\Omega _0+t\Delta)^{-1}\Delta),$
which implies, 
\begin{eqnarray*}
 f'(0)&=&
\tr((S_u-\Sigma_{u0})(\hO-\Omega _0))=\tr(\Omega _0(S_u-\Sigma_{u0})\hO(\Sigma_{u0}-\hSigt))\cr
&=&\sum_{ij}(\Omega _0(S_u-\Sigma_{u0})\hO)_{ij}(\Sigma_{u0}-\hSigt)_{ij}.
\end{eqnarray*}
Note that both $\|\Omega_0\|_1$ and $\|\hO\|_1$ are bounded from above for $\Sigma_{u0},\hSigt\in\Gamma$. By Lemma \ref{la.1}(ii), 
$
\max_{ij}|(\Omega _0(S_u-\Sigma_{u0})\hO)_{ij}|\leq\max_{ij}|(S_u-\Sigma_{u0})_{ij}|\|\Omega_0\|_1\|\hO\|_1=O_p(\sqrt{\log N/T}).
$
Therefore, $|f'(0)|=O_p(\sqrt{\log N/T})\sum_{ij}|\Sigma_{u0, ij}-\hSig_{u, ij}|$. 
 In addition,
$$
f''(t)=\tr((\Omega _0+t\Delta)^{-1}\Delta(\Omega _0+t\Delta)^{-1}\Delta)=vec(\Delta)(\Omega _0+t\Delta)^{-1}\otimes (\Omega _0+t\Delta)^{-1}vec(\Delta),
$$
where $vec$ dentoes the vectorization  operator and $\otimes$ denotes the Kronecker product.  Since both $(\hLamt, \hSigt)$ and $(\Lambda_0,\Sigma_{u0})$ are inside $\Theta_{\lambda}\times\Gamma$, $\sup_{0\leq t\leq 1}\lambda_{\max}(t\hSigti+(1-t)\Sigma_{u0}^{-1})$ is bounded from above, which then implies $
\inf_{0\leq t\leq 1}\lambda_{\min}[(\Omega _0+t\Delta)^{-1}]=\inf_{0\leq t\leq 1}\lambda_{\max}^{-1}(t\hSigti+(1-t)\Sigma_{u0}^{-1})$ is bounded below by a positive constant $c$. Hence $\inf_{0\leq t\leq 1}f''(t)\geq c\|\Delta\|_F^2.$ From (\ref{ea.9}) and $f(1)-f(0)\geq-|f'(0)|+c\|\Delta\|_F^2$,  we have
\begin{eqnarray*}
NQ_1(\hSigt)&\geq&{\mu_T}\sum_{i\neq j} w_{ij}|\hSig_{u, ij}|-\mu_T\sum_{i\neq j} w_{ij}|\Sigma_{u0, ij}|+c\|\Delta\|_F^2-O_p(\sqrt{\frac{\log N}{T}})\sum_{ij}|\Sigma_{u0, ij}-\hSig_{u, ij}|\cr
&=&{\mu_T}\sum_{(i,j)\in S_L} w_{ij}|\hSig_{u, ij}|+\mu_T\sum_{i\neq j, (i,j)\in S_U} w_{ij}|\hSig_{u, ij}|-\mu_T\sum_{i\neq j} w_{ij}|\Sigma_{u0, ij}|+c\|\Delta\|_F^2\cr
&&-O_p(\sqrt{\frac{\log N}{T}})\sum_{\Sigma_{u0,ij}\in S_U}|\Sigma_{u0, ij}-\hSig_{u, ij}|-O_p(\sqrt{\frac{\log N}{T}})\sum_{(i,j)\in S_L}|\Sigma_{u0, ij}-\hSig_{u, ij}|.
\end{eqnarray*}
Since  $|\hSig_{u,ij}|\geq |\hSig_{u,ij}-\Sigma_{u0, ij}|-|\Sigma_{u0, ij}|$, and $\sum_{i\neq j} w_{ij}|\Sigma_{u0, ij}|=\sum_{i\neq j, (i,j)\in S_U} w_{ij}|\Sigma_{u0, ij}|+\sum_{(i,j)\in S_L} w_{ij}|\Sigma_{u0, ij}|$. It follows that
\begin{eqnarray*}
NQ_1(\hSigt)&\geq& 
{\mu_T}\sum_{(i,j)\in S_L} w_{ij}|\hSig_{u, ij}-\Sigma_{u0, ij}|-O_p(\sqrt{\frac{\log N}{T}})\sum_{(i,j)\in S_L}|\Sigma_{u0, ij}-\hSig_{u, ij}|+c\|\Delta\|_F^2\cr
&&-{\mu_T}\sum_{(i,j)\in S_L} w_{ij}|\Sigma_{u0, ij}|-O_p(\sqrt{\frac{\log N}{T}})\sum_{\Sigma_{u0,ij}\in S_U}|\Sigma_{u0, ij}-\hSig_{u, ij}|\cr
&&-{\mu_T}  \sum_{i\neq j, (i,j)\in S_U} w_{ij}[|\Sigma_{u0, ij}|-|\hSig_{u, ij}|]-{\mu_T}  \sum_{(i,j)\in S_L} w_{ij}|\Sigma_{u0, ij}|\cr
&\geq &(\mu_T\min_{(i,j)\in S_L}w_{ij}-O_p(\sqrt{\frac{\log N}{T}}))\sum_{(i,j)\in S_L}|\hSig_{u, ij}-\Sigma_{u0, ij}|+c\|\Delta\|_F^2\cr
&&-2{\mu_T}\sum_{(i,j)\in S_L} w_{ij}|\Sigma_{u0, ij}|-O_p(\sqrt{\frac{\log N}{T}})\sum_{\Sigma_{u0,ij}\in S_U}|\Sigma_{u0, ij}-\hSig_{u, ij}|\cr
&&-\mu_T\max_{i\neq j, (i,j)\in S_U}w_{ij}\sum_{i\neq j, (i,j)\in S_U} |\Sigma_{u0, ij}-\hSig_{u, ij}|\cr
&\geq&\frac{1}{2}\mu_T\min_{(i,j)\in S_L}w_{ij}\sum_{(i,j)\in S_L}|\hSig_{u, ij}-\Sigma_{u0, ij}|+c\|\Delta\|_F^2-2\mu_T\max_{(i,j)\in S_L}w_{ij}K_T\cr
&&-O_p(\sqrt{\frac{\log N}{T}})\sqrt{N+D}\|\Delta\|_F-\mu_T\max_{i\neq j, (i,j)\in S_U}w_{ij}\|\Delta\|_F\sqrt{D},
\end{eqnarray*}
which implies the desired result. 
\end{proof}

\begin{lem}\label{lc.3}
\begin{eqnarray*}
\frac{1}{N}\|\Sigma_{u}-\hSigt\|_F^2&=&O_p\left(\frac{1}{N}\left(\mu_T\max_{(i,j)\in S_L}w_{ij}K_T+\log N+\mu_T^2\max_{i\neq j, (i,j)\in S_U}w_{ij}^2{D}\right)\right)\cr
&&+O_p(\frac{D\log N}{NT}+\sqrt{\frac{\log N}{T}}).
\end{eqnarray*}
\end{lem}
\begin{proof}  Lemma \ref{lc.2} implies
\begin{eqnarray*}
NQ_1(\hSigt)
\geq c\|\Delta\|_F^2-2\mu_T\max_{(i,j)\in S_L}w_{ij}K_T-\left(O_p(\sqrt{\frac{\log N}{T}})\sqrt{N+D}+\mu_T\max_{i\neq j, (i,j)\in S_U}w_{ij}\sqrt{D}\right)\|\Delta\|_F.
\end{eqnarray*}
Lemma \ref{lc.1} gives $
NQ_1(\hSigt)\leq O_p(\log N+N\sqrt{\log N/T}).$
Hence we have
\begin{eqnarray*}
\|\Delta\|_F^2&=&O_p((\sqrt{\frac{(N+D)\log N}{T}}+\mu_T\max_{i\neq j, (i,j)\in S_U}w_{ij}\sqrt{D})^2)\cr
&&+O_p(\mu_T\max_{(i,j)\in S_L}w_{ij}K_T+\log N+N\sqrt{\log N/T})\cr
&=&O_p(\frac{(N+D)\log N}{T}+\mu_T^2\max_{i\neq j, (i,j)\in S_U}w_{ij}^2{D}+\mu_T\max_{(i,j)\in S_L}w_{ij}K_T+\log N+N\sqrt{\log N/T})\cr
&=&O_p(\frac{D\log N}{T}+\mu_T^2\max_{i\neq j, (i,j)\in S_U}w_{ij}^2{D}+\mu_T\max_{(i,j)\in S_L}w_{ij}K_T+\log N+N\sqrt{\log N/T}).
\end{eqnarray*}
Note that $\Sigma_{u0}-\hSigt=\hSigt\Delta\Sigma_{u0}$. Hence the desired result follows from $\|\hSigt\|<M$ wp1 and $\|\Sigun\|<M$.   

\end{proof}

\begin{lem}\label{lc.4}
$N^{-1}\sum_{(i,j)\in S_L}|\hSig_{u, ij}-\Sigma_{u0, ij}|=o_p(1).$
\end{lem}

\begin{proof}

 Lemma \ref{lc.2} implies
\begin{eqnarray*}
 &&\frac{1}{2}\mu_T\min_{(i,j)\in S_L}w_{ij}\sum_{(i,j)\in S_L}|\hSig_{u, ij}-\Sigma_{u0, ij}|\leq NQ_1(\hSigt)+2\mu_T\max_{(i,j)\in S_L}w_{ij}K_T\cr
&&+\left(O_p(\sqrt{\frac{\log N}{T}})\sqrt{N+D}+\mu_T\max_{i\neq j, (i,j)\in S_U}w_{ij}\sqrt{D}\right)\|\Delta\|_F.
\end{eqnarray*}

We have $NQ_1(\hSigt)\leq O_p(\log N+N\sqrt{\log N/T}).$ By Lemma \ref{lc.3},
\begin{eqnarray*}
\|\Delta\|_F&=&O_p(\sqrt{\frac{D\log N}{T}}+\mu_T\max_{i\neq j, (i,j)\in S_U}w_{ij}\sqrt{D})\cr
&&+O_p(\sqrt{\mu_T\max_{(i,j)\in S_L}w_{ij}K_T}+\sqrt{\log N}+\sqrt{N}(\frac{\log N}{T})^{1/4}).
\end{eqnarray*}
which   implies the desired result under  Assumption  4.2. 

\end{proof}

\begin{lem}\label{lc.5} $N^{-1}\Lambda_0' (\hSigti-\Sigma_{u0}^{-1})\Lambda_0 =o_p(1).$
\end{lem}
\begin{proof} Let $\Delta_1=\hSigt-\Sigma_{u0}$, $ \Xi=\Lambda_0' \Sigma_{u0}^{-1}=(\xi_1,...,\xi_N)$, and $\hV=\hSigti\Lambda_0 $. Since the $l_1$ norms of $\hSigti$  and $\Sigma_{u0}^{-1}$ are bounded away from infinity, we have,  $\sup_{i\leq N}\|\hV_i\|=O_p(1)$ and $\sup_{i\leq N}\|\xi_i\|=O(1)$. 
Then
\begin{eqnarray*}
\frac{1}{N}\Lambda_0' (\Sigma_{u0}^{-1}-\hSigti)\Lambda_0 &=&\frac{1}{N} \Xi\Delta_1\hV
=\frac{1}{N}\sum_{(i,j)\in S_L} \xi_i\hV_j'\Delta_{1,ij}+\frac{1}{N}\sum_{\Sigma_{u0, ij}\in S_U} \xi _i\hV_j'\Delta_{1,ij}
\cr
&\leq& O_p(\frac{1}{N})\sum_{(i,j)\in S_L}|\Delta_{1,ij}|+ O_p(\frac{1}{N})\sum_{\Sigma_{u0, ij}\in S_U}|\Delta_{1,ij}|.
\end{eqnarray*}
The first term on the right hand side is $o_p(1)$ by Lemma \ref{lc.4}, and the second is bounded by  $N^{-1}\|\hSigt-\Sigma_{u0}\|\sqrt{N+D}$ (using Cauchy-Schwarz inequality), which is also $o_p(1)$ by Lemma \ref{lc.3} and Assumption    4.2. 

\end{proof}

\begin{lem}\label{lc.6} For $(\hLam,\hSig)=(\hLamt,\hSigt)$, Lemma  \ref{assa.1} is satisfied.

\end{lem}

\begin{proof}
  We first show part (i) of Lemma \ref{assa.1}.  Since $L_c(\hLamt,\hSigt)\leq L_c(\Lambda_0, \Sigma_{u0})$, and $Q_1(\Sigma_{u0})=Q_2(\Lambda_0, \Sigma_{u0})=0$, there is a nonnegative sequence $d_n=O_p(N^{-1}\log N+ T^{-1/2}(\log N)^{1/2})$ such that $Q_1(\hSigt)+Q_2( \hLamt, \hSigt)\leq d_n.$
Lemma \ref{lc.2} then implies $0\leq Q_2(\hSigt,\hLamt)=o_p(1)$. 
On the other hand,
$$
Q_2(\hSigt,\hLamt)=\frac{1}{N}\tr\left[\Lambda_0' \hSigti\Lambda_0 -\Lambda_0' \hSigti\hLamt(\hLam'\hSigti\hLamt)^{-1}
\hLam'\hSigti\Lambda_0 \right].
$$The matrix in the bracket is semi-positive definite.
Hence 
\begin{equation}\label{eqc.2}
\frac{1}{N}\Lambda_0' \hSigti\Lambda_0 -(I_r- J )\frac{1}{N}\hLam'\hSigti\hLamt(I_r- J )'=o_p(1).
\end{equation}
Finally, the desired result follows from  Lemma \ref{lc.5}. 

The first order condition in part (ii) is easy to derive and  is the same as that in Bai and Li (2012).

\end{proof}

\textbf{Proof of Theorem  4.1}

$N^{-1}\|\hSigt-\Sigma_{u0}\|_F^2=o_p(1)$ follows from Lemma \ref{lc.3} and Assumption  4.2.  On the other hand,  equation (\ref{eqc.2})   also implies
$$
\frac{1}{N}(\hLamt-\Lambda_0)'\hSig_u^{-1}(\hLamt-\Lambda_0)-J\frac{1}{N}H^{-1}J'=o_p(1).
$$
By Lemma \ref{la.4}, $N^{-1}JH^{-1}J'=o_p(1)$. Hence $N^{-1}(\hLamt-\Lambda_0)'\hSig_u^{-1}(\hLam-\Lambda_0)=o_p(1)$, which implies the consistency  $N^{-1}\|\hLam-\Lambda_0\|^2=o_p(1)$  because  the  eigenvalues of $\hSig_u^{-1}$ are bounded away from zero. Q.E.D.

To prove the consistency of $\hftt$, we note that the expansion (\ref{eqb.9}) still holds for $\hftt$.  Since $J=o_p(1)$ by Lemma \ref{la.4},  and $\bar{u}$ is of smaller order than $u_t$ for each fixed $t$. Hence $
\hftt-f_t=O_p(N^{-1})\hLamtp\hSigti u_t+o_p(1).
$
Moreover, since $\|\hSigti\|$ and $\|\hSigt\|$ are both $O_p(1)$ and $\|\hLamt\|_F=O_p(\sqrt{N})$ by the restriction of the parameter space $\Theta_{\lambda}\times\Gamma$, we have $N^{-1}\|\hLamtp\hSigti-\Lambda_0\Sigun^{-1}\|_F=O_p(N^{-1/2}\|\hLamt-\Lambda_0\|_F+N^{-1/2}\|\hSigt-\Sigun\|_F)$, which is $o_p(1)$ as proved above. Therefore, since $N^{-1}\Lambda_0'\Sigun^{-1}u_t=N^{-1}\sum_{i=1}^N\xi_iu_{it}=O_p(N^{-1/2})$,
$$
\hftt-f_t=O_p(N^{-1})\Lambda_0'\Sigun^{-1}u_t+o_p(1)=o_p(1).
$$

\subsection{Proof of Theorem  4.2}
We now verify Assumption  4.2 for the Adaptive Lasso.  

\begin{lem}\label{lc.7}
For adaptive lasso,\\
(i) $\min_{i\neq j,  (i,j)\in S_U}|\Sigma_{u0, ij}|^{\gamma}\max_{i\neq j, (i,j)\in S_U}w_{ij}=O_p(1)$.\\
(ii) $\delta_T^{\gamma}\max_{  (i,j)\in S_L}w_{ij}=O_p(1)$,\\
(iii) $\omega_T^{-\gamma}(\min_{  (i,j)\in S_L}w_{ij})^{-1}=O_p(1)$ (recall that $\omega_T=N^{-1/2}+T^{-1/2}(\log N)$).
\end{lem}

\begin{proof} By Lemma \ref{lb.5} 
$\max_{i\leq N, j\leq N}|\hSig_{u,ij}^*-\Sigma_{u0, ij}|=O_p(\tau).$
  Given this result and the assumption that $\min_{(i,j)\in S_U}|\Sigma_{u0, ij}|\gg \omega_T$, we have result (i).  For any $(i,j)\in S_L$, the following inequality holds:
$
\delta_T^{-\gamma}\leq w_{ij}^{-1}\leq (|\Sigma_{u0, ij}|+|\Sigma_{u0, ij}-\hSig_{u, ij}|+\delta_T)^{\gamma},
$
which then implies results (ii) and (iii), due to the assumptions that $\delta_T=o(\omega_T)$, and $\Sigma_{u0, ij}=O(\omega_T).$
\end{proof}

\textbf{Proof of Assumption  4.2 for Adaptive Lasso}

It follows from the previous lemma that
$\alpha_T=O_p(\omega_T^{\gamma}(\min_{i\neq j, \Sigma_{u0, ij}\in S_U}|\Sigma_{u0, ij}|)^{-\gamma})=o_p(1),$
and $\beta_T=O_p((\omega_T/\delta_T)^{\gamma})$. By the assumption that $D=O(N)$,  
$$\zeta=\min\left\{
\sqrt{\frac{T}{\log N}}\frac{N}{D},\left(\frac{T}{\log N}\right)^{1/4}\sqrt{\frac{N}{D}}, \frac{N}{\sqrt{D
\log N}}\right\}\gg\min\left\{
\left(\frac{T}{\log N}\right)^{1/4}, \sqrt{\frac{N}{\log N}}\right\}.$$
 Hence $\alpha_T=O_p(\zeta)$. This together with the lower bound assumption on $\delta_T$ yields Assumption  4.2 (i).

For part (ii), note that $\alpha_T=o_p(1)$ implies that with probability approaching one, $$\min\{N, \frac{N^2}{D}, \frac{N^2}{D}\alpha_T^{-2}\}=N, \hspace{1em}\min\{\frac{N}{D}, \sqrt{\frac{N}{D}}, \frac{N}{D}\alpha_T^{-1}\}=\sqrt{\frac{N}{D}}.$$
By Lemma \ref{lc.7}(ii), (recall that $K_T=\sum_{(i,j)\in S_L}|\Sigma_{u0, ij}|$)  and the lower bound $\delta_T\gg \omega_T(K_T/N)^{1/\gamma}$,  $\mu_T\max_{(i,j)\in S_L}w_{ij}K_T=O_p(\mu_T\delta_T^{-\gamma}K_T)=o_p(N).$

By Lemma \ref{lc.7}(i) and the assumptions that $D=O(N)$ and $\min_{i\neq j, (i,j)\in S_U}|\Sigma_{u0, ij}|\gg \omega_T,$ we have $\mu_T\max_{i\neq j, (i,j)\in S_U}w_{ij}=O_p\mu_T( \min_{i\neq j, (i,j)\in S_U}|\Sigma_{u0, ij}|^{\gamma})^{-1}=o_p(\sqrt{N/D}),
$
due to the upper bound on $\mu_T=o(\omega_T^{\gamma}).$ Finally,  by Lemma \ref{lc.7}(iii) and the assumption  that  $\mu_T\gg \omega_T^{1+\gamma}$,  we have $\mu_T\min_{(i,j)\in S_L}w_{ij}\gg \omega_T.$ 

\textbf{Proof of Assumption  4.2 for SCAD}

Since $\mu_T/\min_{i\neq j, (i,j)\in S_U}|R_{ij}|=o_p(1)$ and $\max_{(i,j)\in S_L}|R_{ij}|=o_p(\mu_T)$,
it is easy to verify that with probability approaching one, 
 $\max_{i\neq j, (i,j)\in S_U}w_{ij}= 0$, $\min_{(i,j)\in S_L}w_{ij}=\max_{(i,j)\in S_L}w_{ij}=\mu_T$. 
   Hence $\alpha_T=0$ and $\beta_T=1.$ This immediately implies the desired result.

\end{document}